\documentclass[12pt]{article}

\usepackage{graphicx}
\usepackage{amsmath}
\usepackage{amsthm}
\usepackage{amssymb}

\usepackage{natbib}

\theoremstyle{definition}
\newtheorem{thm}{Theorem}

\newtheorem{lem}{Lemma}

\theoremstyle{definition}
\newtheorem{defn}{Definition}
\newtheorem{exmp}{Example}

\newcommand{\mbR}{{\mathbb{R}}}
\newcommand{\mbZ}{{\mathbb{Z}}}

\newcommand{\Rd}{\partial}

\newcommand{\cU}{\mathcal{U}}
\newcommand{\cV}{\mathcal{V}}
\newcommand{\cS}{\mathcal{S}}

\begin{document}

\title{
A structural model on a hypercube \\
represented by optimal transport
}
\author{Tomonari SEI}

\maketitle

\begin{abstract}
We propose a flexible statistical model
for high-dimensional quantitative data on a hypercube.
Our model, called the structural gradient model (SGM),
is based on a one-to-one map on the hypercube that is a solution for
an optimal transport problem.
As we show with many examples,
SGM can describe various dependence structures including
correlation and heteroscedasticity.
The maximum likelihood estimation of SGM is effectively solved by
the determinant-maximization programming.
In particular, a lasso-type estimation is available by adding constraints.
SGM is compared with graphical Gaussian models and mixture models.

\noindent
{Keywords:}
 determinant maximization,
 Fourier series,
 graphical model,
 lasso,
 optimal transport,
 structural gradient model.
\end{abstract}

\setlength{\baselineskip}{18pt}

\section{Introduction} \label{section:intro}

In recent years, it becomes more important
to treat high-dimensional quantitative data
especially in biostatistics and spatial-temporal statistics.
The graphical Gaussian model is one of the most important model.
However, the Gaussian model represents only the second-order
interaction without heteroscedasticity.
In this paper, we introduce the structural gradient model (SGM)
that represents both higher-order and heteroscedastic interactions
of data. The model is defined by a transport map that
pushes the target probability density forward to the uniform density.
The data structure is described by the parameters in the transport map.
This model is a practical specification of the gradient model
defined in \cite{sei2006}.

We consider probability density functions
on the hypercube $[0,1]^m$ written as
\begin{equation}
 p(x)\ =\ \det(D^2\psi(x)),
  \quad x\in [0,1]^m,
  \label{eqn:grad}
\end{equation}
where $\psi$ is a convex function and $D^2\psi(x)$
is the Hessian matrix of $\psi$ at $x$.
The function $p$ is a probability density function
if the gradient map $D\psi$ is a bijection on $[0,1]^m$.
In fact, by changing the variable from $x$ to $y=D\psi(x)$,
we obtain
\[
 \int_{[0,1]^m} \det(D^2\psi(x))\mathrm{d}x
 \ =\ \int_{[0,1]^m} \det\left(\frac{\Rd y}{\Rd x}\right)\mathrm{d}x
 \ =\ \int_{[0,1]^m} \mathrm{d}y
 \ =\ 1.
\]
It is known that {\it any} probability density function on $[0,1]^m$
(actually on $\mbR^m$) is written as (\ref{eqn:grad}).
This fact is deeply connected to the theory of optimal transport
(see e.g.\ \cite{villani2003}).
The bijective gradient map $D\psi$, called the Brenier map, is the optimal-transport plan
from the density (\ref{eqn:grad}) to the uniform density.
In this paper, we call $\psi$ {\it the potential function}.
Furthermore, as explained in Section~\ref{section:SGM},
most density functions on $[0,1]^m$ are characterized by
the Fourier series of $\psi$.
When $\psi$ is represented by the Fourier series,
we will call the model (\ref{eqn:grad}) {\it the structural gradient model}
and refer to it as SGM.
Unknown parameters are the Fourier coefficients of the potential function $\psi$.
SGM can describe not only two-dimensional correlations
but also the three-dimensional interactions and heteroscedastic structures,
unlike the graphical Gaussian model.
We examine this flexibility by simulation
and real-data analysis.

The maximum likelihood estimation of SGM
is reduced to a determinant maximization problem
with a robust convex feasible region.
In practice, this region is not directly used
because it is described by infinitely many constraints.
We give two different approaches to overcome this difficulty.
First we give a sequence converging to the feasible region from the inner side.
Secondly we give a $L^1$-conservative region.
These approaches enable us to calculate the estimator
by the determinant maximization algorithm (\cite{vandenberghe1998}).
As a by-product of the second approach we have a lasso-type estimator
for SGM.
A related estimator is the lasso-type estimator for 
graphical Gaussian models
(\cite{meinshausen2006}, \cite{yuan2007}, \cite{bunea2007}, \cite{banerjee2008}).


We consider only the case in which the sample space is a hypercube.
However, this is not a strong assumption because
we can transform any real-valued data into $[0,1]$-valued data by
a fixed sigmoid function.
Unlike the copula models (\cite{nelsen2006}),
the marginal density of SGM does not need to be uniform.
Our model can still adjust the marginal densities
after the sigmoid transform.
Another approach to deal with unbounded data
is given by the author's past papers (\cite{sei2006}, \cite{sei2007}),
where optimal transport between the standard normal density
and other densities is considered.
In this paper, we use the uniform density instead of the normal density
because the former is analytically simpler than the latter.

This paper is organized as follows.
In Section~\ref{section:SGM},
we define SGM and give various examples of it.
In Section~\ref{section:mle},
we investigate the maximum likelihood estimation and
propose a lasso-type estimator.
In Section~\ref{section:numerical},
we compare SGM with graphical Gaussian models
and mixture models by numerical experiments.
Finally we have some discussions in Section~\ref{section:discussion}.
All mathematical proofs are given in Appendix.

\section{The structural gradient model (SGM)} \label{section:SGM}

In this section, we first give the formal definition
and some theoretical properties of SGM.
Then various examples follow.

\subsection{Definition and basic facts} \label{subsection:SGM}

Let $m$ be a fixed positive integer.
Denote the gradient operator on $[0,1]^m$ by $D=(\Rd/\Rd x_i)_{i=1}^m$
and the Hessian operator by $D^2=(\Rd^2/\Rd x_i\Rd x_j)_{i,j=1}^m$.
The determinant of a matrix $A$ is denoted by $\det A$.
The notation $A\succ B$ (resp.\ $A\succeq B$)
means that $A-B$ is positive definite (resp.\ positive semi-definite).
Let $\mbZ_{\geq 0}$ be the set of all non-negative integers.

\begin{defn}[SGM] \label{defn:SGM}
 Let $\cU$ be a finite subset of $\mbZ_{\geq 0}^m$.
 We define {\it the structural gradient model} (abbreviated as {\it SGM}) by
 Eq.~(\ref{eqn:grad}) with the potential function
 \begin{equation}
  \psi(x|\theta)
   \ =\ \frac{1}{2}x^{\top}x
 - \sum_{u\in\cU} \frac{\theta_u}{\pi^2}\prod_{j=1}^m \cos(\pi u_jx_j),
  \label{eqn:SGM}
 \end{equation}
 where $x=(x_j)\in[0,1]^m$
 and $\theta=(\theta_u)\in \mbR^{\cU}$.
 We call $\cU$ {\it the frequency set}.
 The parameter space of SGM is
 \begin{eqnarray}
 \Theta
 &=& 
 \left\{\theta\in\mbR^{\cU} \mid
 D^2\psi(x|\theta)\succeq 0,\quad \forall x\in[0,1]^m
 \right\}.
 \end{eqnarray}
 A vector $\theta\in\mbR^{\cU}$ is called {\it feasible}
 if $\theta\in\Theta$.
 We also call $\Theta$ {\it the feasible region}.
 \qed
\end{defn}

The following lemma is fundamental.

\begin{lem} \label{lem:fundamental}
 If $\theta$ is feasible, then $p(x|\theta)$ is a probability density function on $[0,1]^m$.
\end{lem}

SGM has sufficient flexibility for multivariate modeling
because the following theorem by \cite{caffarelli2000} holds.
To state the theorem, we prepare some notations.
Denote the $2m$ faces of $[0,1]^m$ by
$F_j^b=\{x\in[0,1]^m\mid x_j=b\}$
for $j\in\{1,\ldots,m\}$ and $b\in\{0,1\}$.
For a smooth function $\psi$ on $[0,1]^m$,
we consider a Neumann condition
\begin{eqnarray}
 && \frac{\Rd \psi(x)}{\Rd x_j}
  \ =\ b\ \ \mbox{for\ any}\ \ x\in F_j^b.
  \label{eqn:neumann}
\end{eqnarray}
It is easily confirmed that the function $\psi$ defined by (\ref{eqn:SGM})
satisfies the Neumann condition (\ref{eqn:neumann}).
Conversely, if $\psi(x)$ satisfies the Neumann condition (\ref{eqn:neumann}),
then it is expanded by an infinite cosine series in $L^2$ sense
(see e.g.\ page~300 of \cite{zygmund2002}).
In other words, the function (\ref{eqn:SGM}) approximates
any potential function satisfying (\ref{eqn:neumann})
if we make the frequency set $\cU$ large.
Now we describe the Caffarelli's theorem.
Here we put a slightly stronger assumption than his.

\begin{thm}[Theorem 5 of \cite{caffarelli2000}] \label{thm:caffarelli2000}
 Let $p(x)$ be a strictly positive and continuously differentiable function on $[0,1]^m$.
 Assume that $p(x)$ satisfies a Neumann condition
 $\Rd p(x)/\Rd x_j= 0$ for any $x\in F_j^b$.
 Then there exists a twice-differentiable convex function
 $\psi(x)$ such that (\ref{eqn:grad}) and (\ref{eqn:neumann}) hold.
\end{thm}

Since the conditions for $p(x)$ in the above theorem are
differentiability and a boundary condition, we can construct sufficiently many
statistical models by SGM.
In the following subsection, we enumerate various examples of SGM.
In Section~\ref{section:discussion},
we discuss removal of the boundary condition for $p(x)$
by removing the twice-differentiability condition for $\psi(x)$.

For the one-dimensional case ($m=1$),
SGM becomes a mixture model
as will be explained in the following subsection.
For the multi-dimensional case ($m>1$),
SGM is not a mixture model
except for essentially one-dimensional case.

\begin{lem} \label{lem:SGM-is-not}
 SGM is not a mixture model
 unless there exists some $i\in\{1,\ldots,m\}$
 such that $\cU\subset\mbZ_i$,
 where $\mbZ_i=\{u\in\mbZ_{\geq 0}^m\mid u_j=0\ \forall j\neq i\}$.
\end{lem}

We use the following mixture model as a reference.

\begin{defn}[MixM] \label{defn:mixture}
 Let $\cU$ be a finite subset of $\mbZ_{\geq 0}^m$.
 We define a structural mixture model (referred to as {\it MixM}) by
 \begin{equation}
  \tilde{p}(x|\theta)\ =\ 
  1 + \sum_{u\in\cU} \theta_u\|u\|^2\prod_{j=1}^m \cos(\pi u_jx_j),
  \label{eqn:mixture}
 \end{equation}
 where $x=(x_j)\in[0,1]^m$, $\theta=(\theta_u)\in \mbR^{\cU}$
 and $\|u\|^2=\sum_{j=1}^mu_j^2$.
 The feasible region is $\tilde{\Theta}:=\{\theta\in\mbR^{\cU}\mid \tilde{p}(x|\theta)\geq 0\ \forall x\in [0,1]^m\}$.
 \qed
\end{defn}

In the following lemma,
we prove that SGM and MixM
have a common score function
at the origin $\theta=0$ of the parameter space.
The Fisher information matrix at the origin
is also calculated.

\begin{lem} \label{lem:SGM-and-MixM}
 The score vector at the origin $\theta=0$ of both SGM and MixM
 is equal to
 $(\|u\|^2\prod_{j=1}^m\cos(\pi u_jx_j))_{u\in\cU}$.
 The Fisher information matrix $(J_{uv})_{u,v\in\cU}$ at the origin $\theta=0$ of both the models
 is given by
 \[
  J_{uv}
  \ =\ \frac{\|u\|^41_{\{u=v\}}}{2^{|\sigma(u)|}},
 \]
 where $\sigma(u)=\{j\in\{1,\ldots,m\}\mid u_j>0\}$.
 In particular, $J_{uv}$ is diagonal.
\end{lem}

The Fisher information matrix $J$ at the origin is useful if
we deal with the testing of hypothesis $\theta=0$.
Under this hypothesis, the maximum likelihood estimator $\hat{\theta}$
is approximated by a Gaussian random vector
with mean $0$ and variance $(nJ)^{-1}$.
In Section~\ref{section:numerical},
we will use the scaled maximum likelihood estimator
$J^{1/2}\hat{\theta}$
to detect which components of $\hat{\theta}$
are significant.
A method of computation for the maximum likelihood estimator
is given in Section~\ref{section:mle}.
In general, it seems difficult to calculate the
Fisher information at the other points $\theta\neq 0$.
Exceptional cases will be stated
in the following examples.

\subsection{Examples} \label{subsection:examples}

We enumerate examples of SGM.
We mainly compare SGM with MixM defined in Definition~\ref{defn:mixture}.
For SGM, the following sufficient condition for feasibility of $\theta$ is useful to deal with the examples.
In Theorem~\ref{thm:little},
we will show that $\theta$ is feasible if
\begin{equation}
 1-\sum_{u\in\cU}|\theta_u|u_j^2\ \geq\ 0
  \label{eqn:feasible-intro}
\end{equation}
for any $j=1,\ldots,m$.
This condition is also necessary if, for example,
$\cU$ is a one-element set (see Theorem~\ref{thm:little} for details).

\begin{exmp}[1-dimensional case]
 If $m=1$, then the probability density of SGM is
 given by the Fourier series
 \[
  p(x_1|\theta)
 \ =\ 1+\sum_{u\in\cU}\theta_u u^2\cos(\pi ux_1).
 \]
 This coincides with MixM (Definition~\ref{defn:mixture}).
 The model is considered as a particular case of
 the circular model proposed by \cite{fernandez2004}.
 If $\cU=\{u\}$ with some $u\in\mbZ_{>0}$,
 then the Fisher information $J_{uu}(\theta)$ is explicitly expressed
 for any feasible $\theta=\theta_u$.
 In fact,
 \begin{equation}
  J_{uu}(\theta)
   \ =\ \frac{1-\sqrt{1-\theta^2u^4}}{\theta^2\sqrt{1-\theta^2u^4}}.
   \label{eqn:Fisher-explicit-1}
 \end{equation}
 The proof is given in Appendix.
 \qed
\end{exmp}

\begin{exmp}[Independence]
 Let $m=2$ and 
 \[
  \cU=\{(u_1,0)\mid u_1\in\cU_1\}\cup\{(0,u_2)\mid u_2\in\cU_2\},
 \]
 where $\cU_i$ ($i=1,2$) is a finite subset of $\mbZ_{\geq 0}$.
 Then SGM becomes an independent model
 \[
  p(x_1,x_2|\theta)
 \ =\ \left(1+\sum_{u_1\in\cU_1}\theta_{(u_1,0)}u_1^2\cos(\pi u_1x_1)\right)
 \left(1+\sum_{u_2\in\cU_2}\theta_{(0,u_2)}u_2^2\cos(\pi u_2x_2)\right).
 \]
 Independence of higher-dimensional variables is similarly described.
 On the other hand, if we consider MixM
 \[
  \tilde{p}(x_1,x_2|\theta)
 \ =\ 1+\sum_{u_1\in\cU_1}\theta_{(u_1,0)}u_1^2\cos(\pi u_1x_1)
 +\sum_{u_2\in\cU_2}\theta_{(0,u_2)}u_2^2\cos(\pi u_2x_2),
 \]
 then $x_1$ and $x_2$ are not independent except for trivial cases.
 \qed
\end{exmp}

\begin{exmp}[Correlation]
 Let $m=2$ and $\cU=\{(1,1)\}$.
 Then a pair $(X_1,X_2)$ drawn from $p(x_1,x_2|\theta)$
 has  positive or negative correlation
 if $\theta_{(1,1)}>0$ or $<0$, respectively (see Figure~\ref{fig:correlation-dens}).
 We confirm this observation by explicit calculation.
 We denote $\theta=\theta_{(1,1)}$,
 $c(\xi)=\cos(\pi \xi)$ and $s(\xi)=\sin(\pi \xi)$
 for simplicity.
 The density is
 \begin{eqnarray*}
  p(x_1,x_2|\theta)
   &=& \det\left(
  \begin{array}{cc}
   1+\theta c(x_1)c(x_2)& -\theta s(x_1)s(x_2) \\
   -\theta s(x_1)s(x_2)& 1+\theta c(x_1)c(x_2)
  \end{array}
 \right)
  \\
   &=& 1 + 2\theta c(x_1)c(x_2)
    + \frac{\theta^2}{2}(c(2x_1)+c(2x_2)).
 \end{eqnarray*}
 By the condition (\ref{eqn:feasible-intro}),
 the feasible region for $\theta$ is $[-1,1]$.
 The marginal density of $X_i$ ($i=1,2$) is exactly calculated as
 \[
  p(x_i|\theta)\ =\ 1+\frac{\theta^2}{2}c(2x_i).
 \]
 The mean and variance of $X_i$ ($i=1,2$)
 are $1/2$ and $(1/12)+\theta^2/(4\pi^2)$, respectively.
 The correlation is
 \[
 \frac{\mathrm{Cov}[X_1,X_2]}{\sqrt{\mathrm{V}[X_1]\mathrm{V}[X_2]}}
 \ =\ \frac{8\theta/\pi^4}{(1/12)+\theta^2/(4\pi^2)}
 \ =\ \frac{96\theta/\pi^4}{1+3\theta^2/\pi^2}.
 \]
 The maximum correlation over $\theta\in[-1,1]$ is $96/(\pi^4+3\pi^2)\simeq 0.7558$ at $\theta=1$.
 In contrast, if we consider MixM
 \[
  \tilde{p}(x_1,x_2|\theta)\ =\ 1+2\theta c(x_1)c(x_2),
 \]
 then the feasible region (i.e.\ the set of $\theta$ that assures $\tilde{p}(x_1,x_2|\theta)\geq 0$)
 is $|\theta|\leq 1/2$.
 The correlation is $96\theta/\pi^4$
 and its maximum value is $48/\pi^4\simeq 0.4928$ at $\theta=1/2$.
 Thus SGM can describe a distribution with higher correlation than MixM.
 The Fisher information $J_{uu}(\theta)$
 is explicitly expressed for any feasible $\theta$,
 where $u=(1,1)$.
 The formula is
 \begin{equation}
  J_{uu}(\theta)\ =\ \frac{2(1-\sqrt{1-\theta^2})}{\theta^2\sqrt{1-\theta^2}}.
   \label{eqn:Fisher-explicit-2}
 \end{equation}
 The proof is given in Appendix.
 \begin{figure}[htbp]
  \begin{center}
   \begin{tabular}{cc}
   \includegraphics[width=7cm,clip]{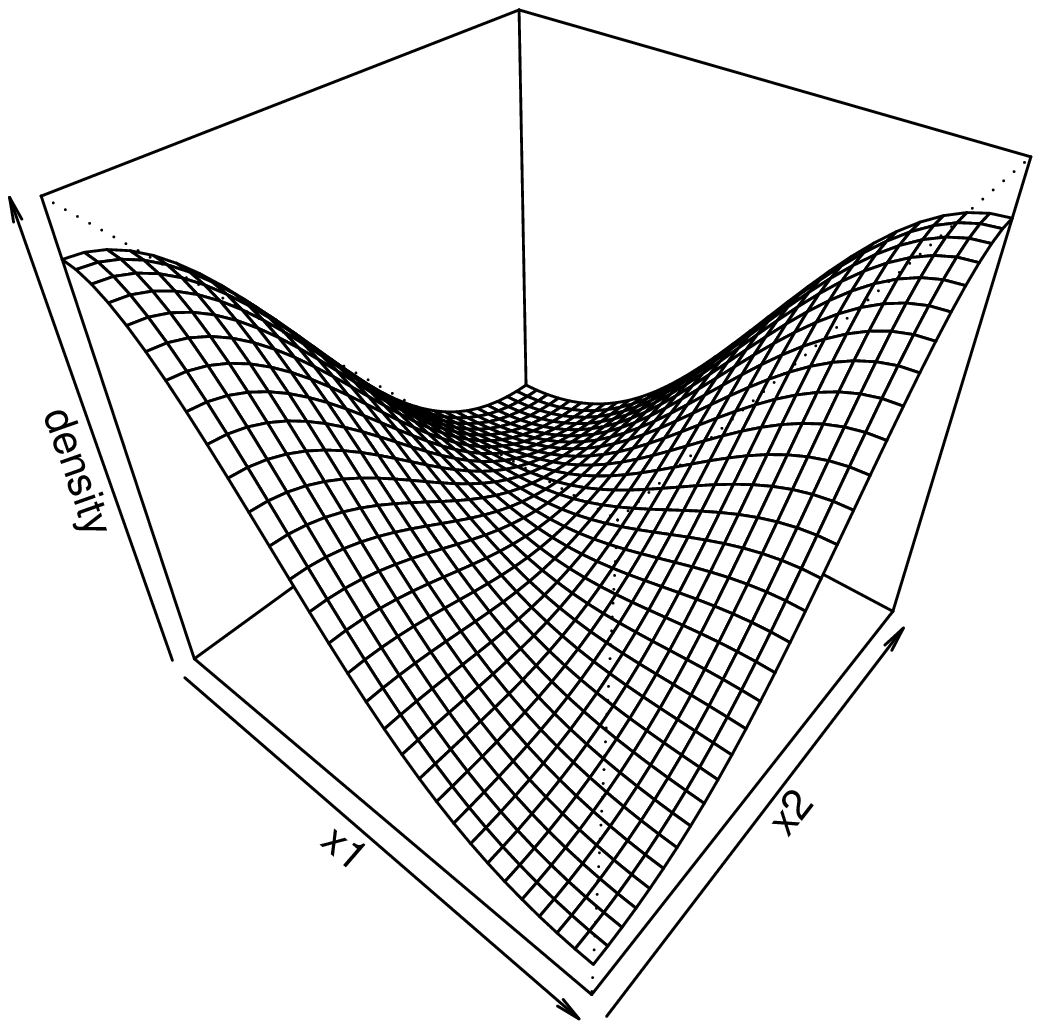}
    &
    \includegraphics[width=7cm,clip]{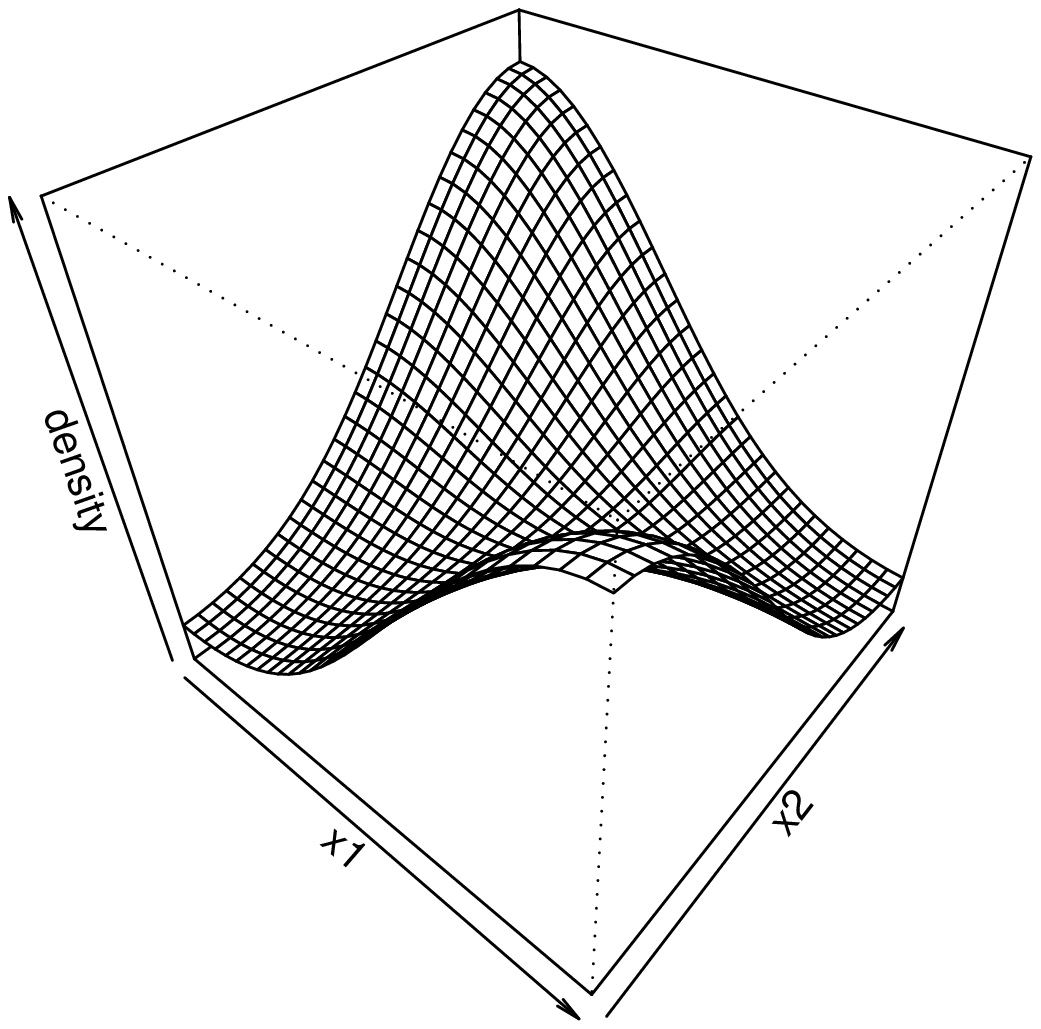}
    \\
    (a) $\theta=0.5$.
    &
    (b) $\theta=-0.5$.
   \end{tabular}
  \end{center}
  \caption{The probability density $p(x|\theta)$ for $\cU=\{(1,1)\}$ and
  $\theta=\theta_{(1,1)}=\pm 0.5$.
  The correlation coefficient is about $\pm 0.458$ for $\theta=\pm 0.5$,
  respectively.}
  \label{fig:correlation-dens}
 \end{figure}
 \qed
\end{exmp}

\begin{exmp}[Heteroscedasticity] \label{exmp:hetero}
 Let $m=2$ and $\cU=\{(1,2)\}$.
 Then a pair $(X_1,X_2)$ drawn from $p(x_1,x_2|\theta)$
 has the following property:
 the conditional mean of $X_2$ given $X_1$ does not depend on $X_1$
 but the conditional variance does (see Figure~\ref{fig:hetero-dens}).
 In other words, $X_2$ has heteroscedasticity
 in terms of regression analysis.
 We confirm this fact.
 The joint density is
\begin{eqnarray*}
 p(x_1,x_2|\theta)
 &=& \det\left(
 \begin{array}{cc}
  1+\theta c(x_1)c(2x_2)& -2\theta s(x_1)s(2x_2) \\
  -2\theta s(x_1)s(2x_2)& 1+4\theta c(x_1)c(2x_2) \\
 \end{array}
 \right)
 \\
 &=& 1 +5\theta c(x_1)c(2x_2)
  + 2\theta^2 c(2x_1)
  +2\theta^2c(4x_2)
\end{eqnarray*}
 where we put $c(\xi)=\cos(\pi\xi)$, $s(\xi)=\sin(\pi\xi)$,
 and $\theta=\theta_{(1,2)}$.
 The marginal density of $X_1$ is
 $p(x_1)=1+2\theta^2 c(2x_1)$.
 The conditional density of $X_2$ given $X_1$ is
\begin{eqnarray*}
  p(x_2|x_1,\theta)
 &=& 1+\frac{5\theta c(x_1)c(2x_2)
  +2\theta^2c(4x_2)}
 {1+2\theta^2 c(2x_1)}
\end{eqnarray*}
 The conditional mean of $X_2$ given $X_1$ is exactly $1/2$,
 and therefore the correlation between $X_1$ and $X_2$ is zero.
 However, the conditional variance of $X_2$ given $X_1$ is
 not constant:
 \[
  \int_0^1 (x_2-1/2)^2p(x_2|x_1,\theta)\mathrm{d}x_2
 \ =\ \frac{1}{12}+
 \frac{10\theta c(x_1)+\theta^2}
 {4\pi^2\{1+2\theta^2c(2x_1)\}}.
 \]
 In order to measure the dependency of $X_1$, 
 let us consider the quantity
 \begin{eqnarray*}
  \beta_{122}(\theta) &=&
 \frac{\mathrm{E}[(X_1-1/2)(X_2-1/2)^2]}
 {\{\mathrm{V}[X_1]\}^{1/2}\mathrm{V}[X_2]}.
  \\
 &=& \frac{-5\theta/\pi^4}{\{(1/12)+\theta^2/\pi^2\}^{1/2}\{(1/12)+\theta^2/(4\pi^2)\}}
 \end{eqnarray*}
 The maximum value of $\beta_{122}(\theta)$ over the feasible region $\theta\in[-1/4,1/4]$
 is $\beta_{122}(-1/4)\simeq 0.5047$.
 In contrast, for MixM
 $\tilde{p}(x_1,x_2|\theta)=1+5\theta c(x_1)c(2x_2)$,
 the maximum of $\beta_{122}(\theta)$ over the feasible region $\theta\in[-1/5,1/5]$
 is $\tilde{\beta}_{122}(-1/5)\simeq 0.4267$.
 Thus SGM can describe more heteroscedastic distributions than
 MixM.
 The heteroscedasticity appears in regression analysis,
 where explanatory and response variables are {\it a priori} selected.
 Remark that our model does not need a priori selection of variables.
 \qed
 \begin{figure}[htbp]
  \begin{center}
   \begin{tabular}{c}
   \includegraphics[width=7cm,clip]{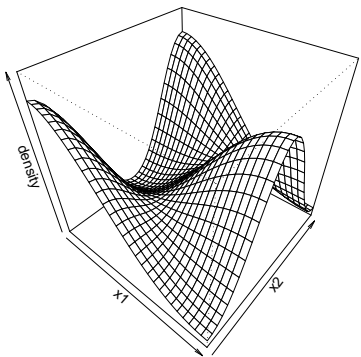}
   \end{tabular}
  \caption{The probability density for $\cU=\{(1,2)\}$ and $\theta=0.2$.
  The conditional density $p(x_2|x_1)$ is unimodal
  if $x_1$ is close to $1$, and bimodal if $x_1$ is close to $0$.}
  \label{fig:hetero-dens}
  \end{center}
 \end{figure}
\end{exmp}

\begin{exmp}[three-dimensional interaction]
 Let $m=3$ and $\cU=\{(1,1,1)\}$.
 Then the triplet $(X_1,X_2,X_3)$ has the three-dimensional
 interaction although the marginal two-dimensional correlation for any pair vanishes.
 We confirm this.
 The joint probability density is
\begin{eqnarray*}
  p(x_1,x_2,x_3|\theta)
 &=& 1+3\theta c_1c_2c_3+3\theta^2c_1^2c_2^2c_3^2
 +\theta^3c_1^3c_2^3c_3^3
 \\
 && -2\theta^3c_1s_1^2c_2s_2^2c_3s_3^2
  -(1+\theta c_1c_2c_3)\theta^2(c_1^2s_2^2s_3^2+s_1^2c_2^2s_3^2+s_1^2s_2^2c_3^2),
\end{eqnarray*}
 where $c_i=\cos(\pi x_i)$ and $s_i=\sin(\pi x_i)$ for $i=1,2,3$.
 The density is symmetric with respect to permutation of axes.
 The feasible region is $|\theta|\leq 1$ by (\ref{eqn:feasible-intro}).
 The 2-dimensional and 1-dimensional marginal densities are
$p(x_1,x_2|\theta)=1+\theta^2(4c_1^2c_2^2-1)/2$
 and 
 $p(x_1|\theta)=1+\theta^2(2c_1^2-1)/2$,
 respectively.
 In particular, the mean of $X_i$ is $1/2$
 and the correlation of $X_i$ and $X_j$ ($i\neq j$) is zero.
 However, there exists three-dimensional interaction between $(X_1,X_2,X_3)$.
 We calculate
 \[
 \beta_{123}(\theta)\ :=\ 
 \frac{\mathrm{E}[(X_1-\mathrm{E}X_1)(X_2-\mathrm{E}X_2)(X_3-\mathrm{E}X_3)]}
 {\sqrt{\mathrm{V}[X_1]\mathrm{V}[X_2]\mathrm{V}[X_3]}}.
 \]
 The result is
 \[
 \beta_{123}(\theta)
 \ =\ \frac{-24\theta/\pi^6-1944\theta^3/729\pi^6}{(1/12+\theta^2/(4\pi^2))^{3/2}}.
 \]
 The maximum value of $\beta_{123}(\theta)$ over the feasible region $|\theta|\leq 1$
 is $\beta_{123}(-1)\simeq 0.7743$.
 In contrast, for MixM $\tilde{p}(x_1,x_2,x_3|\theta)=1+3\theta c_1c_2c_3$,
 we have $\beta_{123}(\theta)=-288\sqrt{12}\theta/\pi^6$.
 Its maximum value over the feasible region $|\theta|\leq 1/3$
 is about $0.3459$ at $\theta=-1/3$.
 \qed
\end{exmp}

\begin{exmp}[Approximately conditional independence]
 Let $m=3$ and $(X_1,X_2,X_3)$ be drawn from a probability density $p(x_1,x_2,x_3)$.
 In general, conditional independence of $X_1$ and $X_2$
 given $X_3$ is described by 
 $p(x_1,x_2,x_3)=p(x_3)p(x_1|x_3)p(x_2|x_3)$
 or, equivalently, the conditional mutual information
 \[
  I_{12|3}
  \ =\ \int p(x_1,x_2,x_3)\log \frac{p(x_1,x_2|x_3)}{p(x_1|x_3)p(x_2|x_3)}
  \mathrm{d}x_1\mathrm{d}x_2\mathrm{d}x_3
 \]
 vanishes.
 A log-linear model $\exp(f(x_1,x_3)+g(x_2,x_3))$
 satisfies this condition.
 Although SGM does not represent
 any conditional-independence model,
 we can construct an approximately conditional-independence model.
 Let $m=3$ and $\cU=\{(1,0,1),(0,1,1)\}$.
 Then, by putting $c_i=\cos(\pi x_i)$, $s_i=\sin(\pi x_i)$, $\theta=\theta_{(1,0,1)}$
 and $\phi=\theta_{(0,1,1)}$, we have
\begin{eqnarray*}
 \lefteqn{p(x_1,x_2,x_3|\theta,\phi)}
  \\
  &=& \det
  \left(
   \begin{array}{ccc}
   1 + \theta c_1c_3 & 0 & -\theta s_1s_3 \\
   0 & 1 + \phi c_2c_3 & -\phi s_2s_3 \\
   -\theta s_1s_3 & -\phi s_2s_3 & 1+\theta c_1c_3+\phi c_2c_3
   \end{array}
  \right)
  \\
  &=& 1+2\theta c_1c_3+2\phi c_2c_3+3\theta\phi c_1c_2c_3^2
  +\theta^2(c_1^2c_3^2-s_1^2s_3^2)+\phi^2(c_2^2c_3^2-s_2^2s_3^2)
  \\
  && +\theta^2\phi (c_1^2c_3^2-s_1^2s_3^2)c_2c_3+\theta\phi^2(c_2^2c_3^2-s_2^2s_3^2)c_1c_3
\end{eqnarray*}
 Now assume that $\epsilon:=\max(|\theta|,|\phi|)$ is close to zero.
 Then the conditional mutual information is,
 after tedious calculations,
 \[
 I_{12|3}
 \ =\ \frac{3}{16}\theta^2\phi^2 + {\rm O}(\epsilon^5).
 \]
 On the other hand,
 MixM
 $\tilde{p}(x_1,x_2,x_3|\theta,\phi)=1+2\theta c_1c_3+2\phi c_2c_3$
 has the conditional mutual information $I_{12|3}=(3/4)\theta^2\phi^2+{\rm O}(\epsilon^5)$.
 The leading term is 4 times larger than that of SGM.
 \qed
\end{exmp}

We summarize the above examples in Table~\ref{table:exmp}.

\begin{table}[htbp]
 \caption{Summary of the examples.
 For each example, the characteristics of SGM and MixM are compared.}
 \label{table:exmp}
 \begin{center}
  \begin{tabular}{|c|l|c|l||c|c|}
   \hline
   \#& \multicolumn{1}{|c|}{Model name} & $m$ & 
        \multicolumn{1}{|c||}{Characteristic} & SGM & MixM \\
   \hline
   1& 1-dim. & 1 & (SGM$=$MixM) & --- & --- \\
   2& independence & 2 & `is independent' & TRUE & FALSE \\
   3& correlation & 2& maximum correlation & 0.7558 & 0.4928 \\
   4& heteroscedasticity & 2 & maximum $\beta_{122}$ & 0.5047 & 0.4267 \\
   5& 3-dim.~interaction & 3 & maximum $\beta_{123}$ & 0.7743 & 0.3459 \\
   6& conditional independence & 3 & leading coefficient of $I_{12|3}$ & $3/16$ & $3/4$ \\
   \hline
  \end{tabular}
 \end{center}
\end{table}

\begin{exmp} \label{exmp:complicated}
 We can construct more complicated densities
 by combining the preceding ones.
 For example, let $m=3$ and
 $\cU=\{(1,2,0),(0,1,1),(1,1,1)\}$.
 Let the corresponding parameter vector be
 $\theta=(0.1,0.3,0.2)$.
 The vector $\theta$ is feasible since (\ref{eqn:feasible-intro})
 is satisfied.
 The marginal and conditional 2-dimensional densities are illustrated in Figure~\ref{fig:high-1}.
 \qed
\end{exmp}

\begin{figure}[htbp]
 \begin{tabular}{ccc}
  \includegraphics[width=4.5cm]{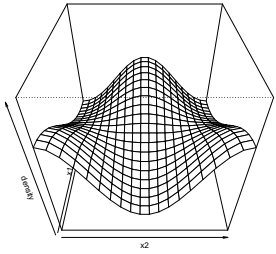}
  & \includegraphics[width=4.5cm]{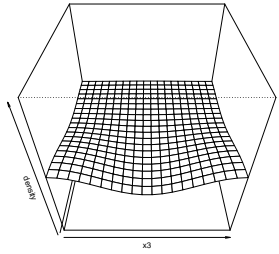}
  & \includegraphics[width=4.5cm]{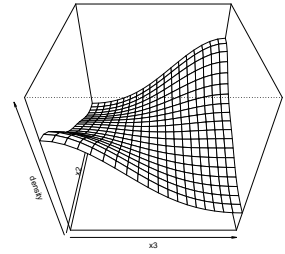}
  \\
  (a) $p(x_1,x_2)$ & (b) $p(x_1,x_3)$ & (c) $p(x_2,x_3)$
  \\
  \includegraphics[width=4.5cm]{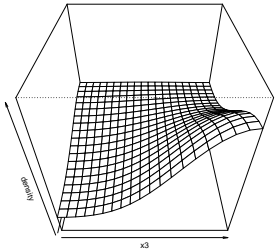}
  & \includegraphics[width=4.5cm]{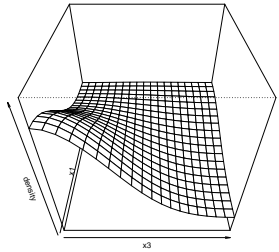}
  &
  \\
  (d) $p(x_1,x_3|x_2=3/4)$ & (e) $p(x_1,x_3|x_2=1/4)$ &
 \end{tabular}
 \caption{The marginal and conditional densities
 for $\cU=\{(1,2,0),(0,1,1),(1,1,1)\}$.
 The figures (a), (b) and (c) are the marginal density
 $p(x_i,x_j)$ for each pair $(i,j)$.
 The figures (d) and (e) are
 the conditional density $p(x_1,x_3|x_2)$
 for specific values of $x_2$.}
 \label{fig:high-1}
\end{figure}

%
%
%

\section{Maximum likelihood estimation of SGM} \label{section:mle}

Let $x(1),\ldots,x(n)$ be independent samples drawn from the
true density $p_0(x)$ whose support is $[0,1]^m$.
From the definition of SGM,
the maximum likelihood estimation of SGM is formulated as
a convex optimization program:
\begin{eqnarray*}
 \mathrm{maximize} && \sum_{t=1}^n
  \log\det\left(I+\sum_{u\in\cU}\theta_uH_u(x(t))\right),\\
 \mathrm{subject\ to} &&
  \theta\in\Theta=\left\{
 \theta\in\mbR^{\cU}\ \left|\ 
  I+\sum_{u\in\cU}\theta_uH_u(\xi)\succeq 0\quad \forall\xi\in [0,1]^m
 \right.
 \right\},
\end{eqnarray*}
where we put $H_u(x)=D^2(-\pi^{-2}\prod_{\rho=1}^m\cos(\pi u_{\rho}x_{\rho}))$.
Recall that $D^2$ is the Hessian operator and $\cU$ is a finite subset of $\mbZ_{\geq 0}^m$.

It is hard to write down $\Theta$ explicitly.
The difficulty follows from the statement ``for any $\xi\in[0,1]^m$''
in the definition of $\Theta$.
In general, for a set of feasible regions $\Theta_{\alpha}$
indexed by $\alpha$, the region $\cap_{\alpha}\Theta_{\alpha}$
is called a robust feasible region (see \cite{ben-tal1998}).

We consider two approaches to solve this problem.
We will first give a sequence $\Theta_M^{\circ}$ of regions
converging to $\Theta^{\circ}$, the interior of $\Theta$, as $M\to\infty$.
Hence the maximum likelihood estimator is calculated
with arbitrary accuracy in principle.
However, $\Theta_M^{\circ}$ has about $M^m$ constraints on $\theta$ and therefore
it is usually expensive if $m\geq 3$.
For the second approach, we give a proper subset
$\Theta^{\rm lit}$ of $\Theta$,
which consists of only $m$ constraints.
As a by-product of the second approach,
we obtain a lasso-type estimator
because $\Theta^{\rm lit}$ is compatible with $L^1$-constraints.
We call the maximizer of the log-likelihood
over these constrained regions
{\it the constrained maximum likelihood estimator}.
The constrained maximum likelihood estimator is calculated via
the determinant maximization algorithm
(\cite{vandenberghe1998}).

If $m=1$, the feasible region is the set of Fourier coefficients
of non-negative functions.
To deal with the feasible region,
\cite{fernandez2004} used Fej\'er's characterization:
the Fourier series of any non-negative function
is written as the square of a Fourier series.
More specifically, for any $r(x)=\sum_{u=0}^{\infty}r_u\cos(\pi ux)$,
its square $r(x)^2$ is of course non-negative and
written by a Fourier series.
The Fourier coefficients of $r(x)^2$ are written by
quadratic polynomials of $(r_u)_{u=0}^{\infty}$.
However, it is hard to use this representation for our problem
because we assume $\theta_u=0$ for $u\notin\cU$
and this restriction is not affine in $r_u$.

\subsection{Inner approximation of feasible region}

Let $\Theta^{\circ}$ be the interior of $\Theta$.
We give a sequence of tractable sets
$\Theta_M^{\circ}$ that converges to $\Theta^{\circ}$
from inside as $M\to\infty$.
We first remark the following lemma.

\begin{lem} \label{lem:interior}
 The set $\Theta^{\circ}$ is equal to
 $\left\{\theta\in\mbR^{\cU} \mid 
  D^2\psi(x|\theta)\succ 0\ \ \forall x\in[0,1]^m
 \right\}$.
\end{lem}

We prepare some notations for constructing $\Theta_M^{\circ}$.
We consider the lattice points $L_M^m$,
where $L_M=\{\frac{0}{M},\frac{1}{M},\cdots,\frac{M}{M}\}$.
Let $\|u\|_{\infty}=\max_{j}|u_j|$
and $U_{\max}=\max_{u\in\cU}\|u\|_{\infty}$.
Define a linear operator $K_M$ on $\mbR^{\cU}$
by $(K_M\theta)_u=\theta_u/\prod_{j=1}^m(1-u_j/M)$ for $\theta\in\mbR^{\cU}$.
Finally, we define $\Theta_M^{\circ}$ for each $M\geq U_{\max}+1$ by
\[
 \Theta_M^{\circ}
 \ =\ \left\{ \theta\in\mbR^{\cU}
 \ \left|\ 
 D^2\psi(\xi|K_M\theta)\succ 0,\ \forall \xi\in L_M^m
 \right.
 \right\}.
\]
Remark that $\Theta_M^{\circ}$
is written in a finite number of constraints,
in contrast to $\Theta^{\circ}$ and $\Theta$.
We have the following theorem.

\begin{thm} \label{thm:inner-approx}
 For any $M\geq U_{\max}+1$, we have $\Theta_M^{\circ}\subset \Theta^{\circ}$
 and
 \[
  \Theta^{\circ}\ =\ 
  \limsup_{M\to\infty}\Theta_M^{\circ},
 \]
 where $\limsup_{M\to\infty}\Theta_M^{\circ}$
 is defined by $\cap_{M'\geq 1}\cup_{M\geq M'}\Theta_M^{\circ}$.
\end{thm}

The constrained maximum likelihood estimator
of $\theta$ over $\Theta_M^{\circ}$ is calculated via
the determinant maximization algorithm
(\cite{vandenberghe1998}).
Hence, in principle, we can calculate
the maximum likelihood estimator with
arbitrary accuracy.
However, the region $\Theta_M^{\circ}$
consists of $|L_M^m|=(M+1)^m$ constraints.
This number is usually expensive if $m\geq 3$.
In the following subsection,
we give a proper subset of $\Theta$
which consists of only $m$ constraints.

\begin{exmp}
 Let $m=2$ and $\cU=\{(1,1),(2,2)\}$.
 The approximated regions $\Theta_M^{\circ}$ ($M=5,10,20,40$)
 are illustrated in Figure~\ref{fig:2dim-feas} (a).
 For this case, we can give a precise expression of $\Theta$.
 The two eigenvalues of the Hessian matrix $D^2\psi(x|\theta)$
 are given by
 \begin{eqnarray*}
  \lambda_{\pm} &=& 1+\theta_{(1,1)}\cos(\pi(x_1\pm x_2))+4\theta_{(2,2)}\cos(2\pi(x_1\pm x_2)).
 \end{eqnarray*}
 In the theory of time-series analysis,
 the function $f(z):=1+\sum_k\rho_k\cos(kz)$ of $z$
 is the spectral density of a MA($k$) process
 with the autocorrelation coefficients $(\rho_j)_{j=1}^k$.
 In particular, for MA(2),
 it is known that $f(z)$ is non-negative for any $z$
 if and only if
 $|\rho_1|+|\rho_2|\leq 1$ or $\rho_1^2\leq 4\rho_2(1-\rho_2)$ holds
 (see \cite{box1976}, Section 3.4).
 Therefore the feasible region for $\cU=\{(1,1),(2,2)\}$ is given by
 \[
  |\theta_{(1,1)}|+|4\theta_{(2,2)}|
 \ \leq\ 1\quad \mbox{or}\quad
 (\theta_{(1,1)})^2\leq 4(4\theta_{(2,2)})(1-4\theta_{(2,2)}).
 \]
 The region $\Theta_M^{\circ}$ shown in Figure~\ref{fig:2dim-feas} (a)
 is close to this region.
 We also illustrate the approximated regions
 for another example $\cU=\{(1,1),(3,1)\}$ in Figure~\ref{fig:2dim-feas} (b).
 \qed
\begin{figure}[htbp]
\begin{center}
\begin{tabular}{cc}
\includegraphics[width=6cm]{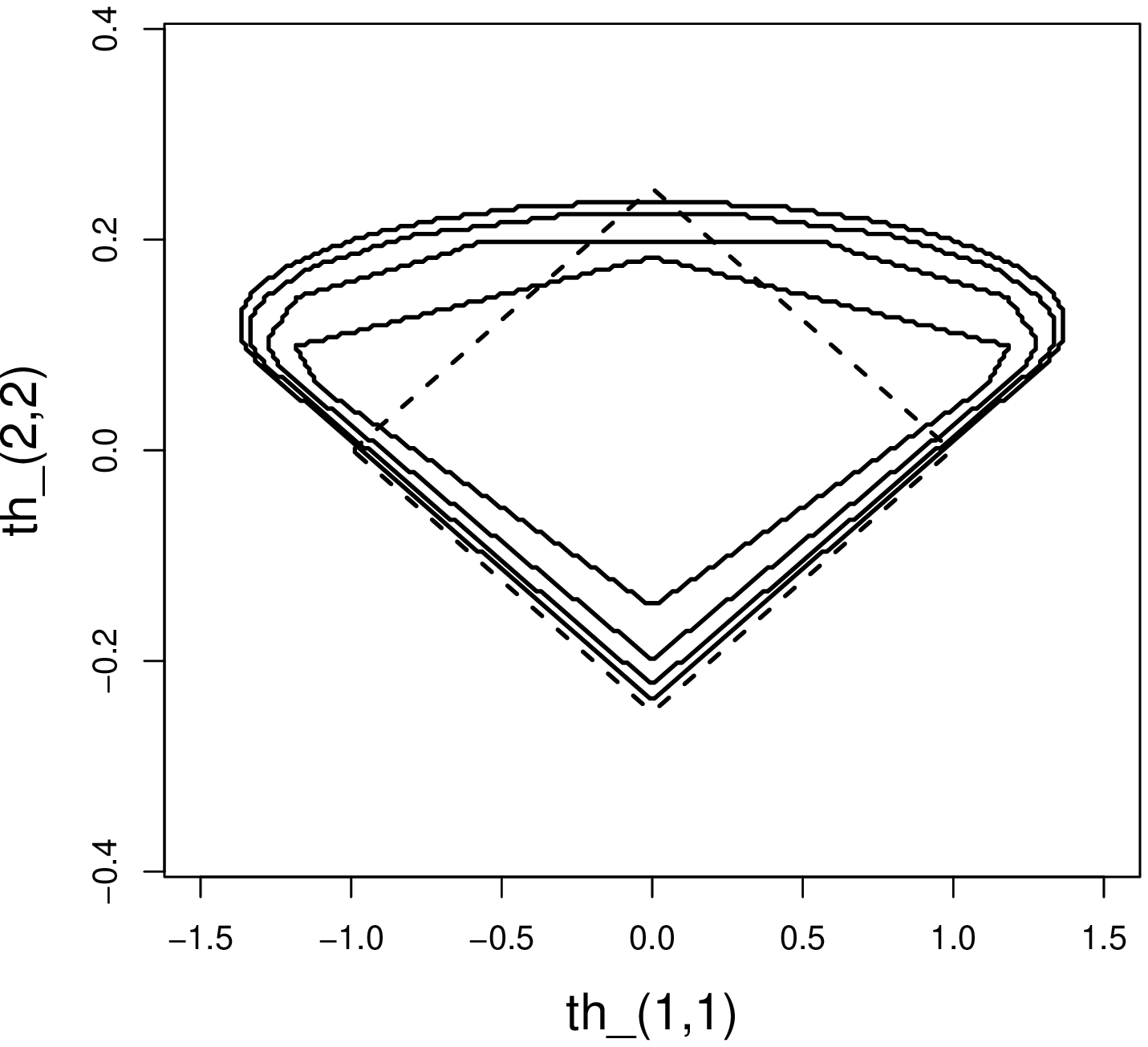}
 &
\includegraphics[width=6cm]{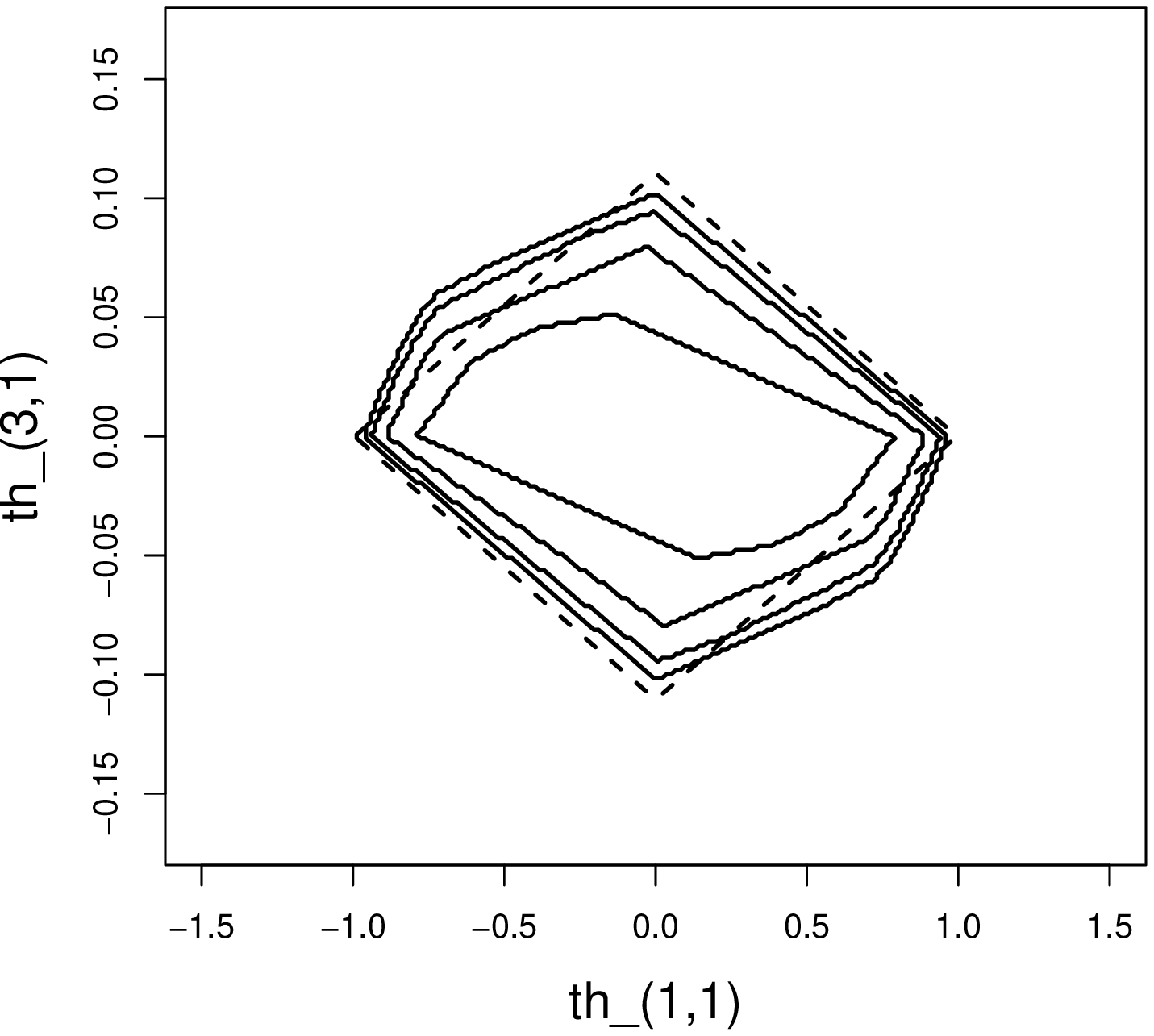}
 \\
 (a)\ $\cU=\{(1,1),(2,2)\}$.
 &
 (b)\ $\cU=\{(1,1),(3,1)\}$.
\end{tabular}
\end{center}
 \caption{The approximated region $\Theta_M^{\circ}$ (solid line;
 $M=5,10,20,40$ from inner side)
 and the little parameter space $\Theta^{\rm lit}$ (dashed line)
 defined in Subsection~\ref{subsection:little}.}
 \label{fig:2dim-feas}
\end{figure}
\end{exmp}

We remark that
the feasible region for MixM (Definition~\ref{defn:mixture})
is approximated from the inner side by
\[
 \tilde{\Theta}_M^{\circ}\ :=\ \left\{\theta\in\mbR^{\cU}\mid
  \tilde{p}(\xi|K_M\theta)>0,\ \mbox{for}\ \xi\in L_M^m\right\}
\]
The proof is similar to that of Theorem~\ref{thm:inner-approx} and omitted here.

\subsection{A conservative region and Lasso-type estimation} \label{subsection:little}

We give a sufficient condition such that $\theta\in\Theta$.
Define a set $\Theta^{\rm lit}$ by
\[
 \Theta^{\rm lit}
 \ =\ \left\{
 \theta\in\mbR^{\cU}
 \ \left|\ 
 1 - \sum_{u\in\cU}|\theta_u|u_j^2\geq 0\quad
 (\forall j=1,\ldots,m)
 \right.
 \right\}.
\]
We call $\Theta^{\rm lit}$ {\it the little parameter space}.
It is an intersection of $m$ constraints.
In the following theorem, we show that the little parameter space $\Theta^{\rm lit}$
is a subset of the feasible region $\Theta$.
In other words, $\Theta^{\rm lit}$ is more conservative than $\Theta$ in the sense of robustness.
We say that a subset $\cV$ of $\cU$ is linearly independent modulo 2
if a linear map $\ell:\{0,1\}^{\cV}\mapsto \{0,1\}^m$
defined by $\ell(\epsilon)=\sum_{u\in\cV}\epsilon_uu$ (mod $2$)
has the kernel $\{0\}$.
For each $\cV\subset\cU$,
the set of vectors that have only $\cV$-components
is denoted by $\mbR_{\cV}=\{\theta\in\mbR^{\cU}\mid \theta_u=0\ \forall u\notin\cV\}$.

\begin{thm} \label{thm:little}
 For any $\cU$, $\Theta^{\rm lit}\subset\Theta$.
 Furthermore, if a subset $\cV$ of $\cU$ is linearly independent modulo 2,
 then we have $\Theta^{\rm lit}\cap\mbR_{\cV}=\Theta\cap\mbR_{\cV}$.
 In particular, if $\cU$ itself is linearly independent modulo 2,
 then $\Theta^{\rm lit}=\Theta$.
\end{thm}

By letting $\cV$ be a one-element set $\{u\}$,
we have the relation $\Theta^{\rm lit}\cap \mbR_{\{u\}}=\Theta\cap\mbR_{\{u\}}$.
This shows that $\Theta^{\rm lit}$
contains at leat $2|\cU|$ boundary points of $\Theta$.
The little parameter space
for $\cU=\{(1,1),(2,2)\}$ and $\cU=\{(1,1),(3,1)\}$
is indicated in Figure~\ref{fig:2dim-feas} (a) and (b), respectively.

The constrained maximum likelihood estimator of $\theta$ over $\Theta^{\rm lit}$
is computed via the determinant maximization algorithm
by introducing non-negative slack variables $\theta_u^+$ and $\theta_u^-$
such that $\theta_u=\theta_u^+-\theta_u^-$ and $|\theta_u|=\theta_u^++\theta_u^-$.
The estimator is usually sparse.
This sparsity is closely related to the lasso estimator
\cite{tibshirani1996} in that
the regression method
is executed with $L^1$-constraints.
Our little parameter space $\Theta^{\rm lit}$
is also represented by $L^1$-constraints.
Hence we call the constrained maximum likelihood estimator of $\theta$ over $\Theta^{\rm lit}$
{\it the lasso-type estimator for SGM}.
Furthermore, we will use an indexed set $\Theta_{\tau}^{\rm lit}$ with
a tuning parameter $\tau\in [0,1]$ by
\[
 \Theta_{\tau}^{\rm lit}
 \ =\ \left\{
 \theta\in\mbR^{\cU}
 \ \left|\ 
  \tau - \sum_{u\in\cU}|\theta_u|u_j^2\geq 0\quad
  (\forall j=1,\ldots,m)
 \right.
 \right\}.
\]
In particular, $\Theta_0^{\rm lit}=\{0\}$
and $\Theta_1^{\rm lit}=\Theta^{\rm lit}$.
The tuning parameter $\tau$ can be selected
by cross validation.

We remark that
the feasible region for MixM (Definition~\ref{defn:mixture})
has the following conservative region
\[
 \tilde{\Theta}^{\rm lit}\ :=\ 
 \left\{
  \theta\in\mbR^{\cU}\ \left|\ 
  1-\sum_{u\in\cU}|\theta_u|\|u\|^2\geq 0
 \right.
 \right\}.
\]
Furthermore, if a subset $\cV$ of $\cU$ is linearly
independent modulo 2, then we have $\tilde{\Theta}^{\rm lit}\cap\mbR_{\cV}
=\tilde{\Theta}\cap\mbR_{\cV}$.
The proof is similar to that of Theorem~\ref{thm:little} and is omitted here.

Recently, lasso-type estimators for graphical Gaussian models
are proposed by several authors:
\cite{yuan2007}, 
\cite{banerjee2008} 
and \cite{friedmann2008}.
On the other hand, a sparse density estimation (SPADES)
for mixture models is considered in \cite{bunea2007}.
Our MixM is considered as a version of SPADES
although the estimation procedure is different.
In Section~\ref{section:numerical},
we compare SGM with 
MixM and the graphical Gaussian model by numerical examples.

\section{Numerical examples} \label{section:numerical}

We give numerical examples on simulated and real datasets.
We calculate the constrained maximum likelihood estimator
and study its predictive performance.
We compare SGM with the graphical Gaussian model (with lasso)
and MixM (Definition~\ref{defn:mixture}).

We describe some notations and assumptions.
We use the following frequency set for SGM throughout this section:
\begin{equation}
 \cU\ =\ \left\{u\in\mbZ_{\geq 0}^m
 \ \left|\ \|u\|_{\infty}\leq 2,\ \|u\|_{1}\leq 3\right.
 \right\},
  \label{eqn:numerical-freq}
\end{equation}
where $\|u\|_{\infty}=\max_j|u_j|$ and $\|u\|_1=\sum_j|u_j|$.
The elements of $\cU$ are given by
$(1,0,\ldots,0)$,
$(2,0,\ldots,0)$,
$(1,1,0,\ldots,0)$,
$(2,1,0,\ldots,0)$,
$(1,1,1,0,\ldots,0)$
and their permutations of the components.
The cardinality of $\cU$ is
$m(m+1)(m+5)/6$.
Let $\hat{\theta}_M^{\circ}=(\hat{\theta}_{M,u}^{\circ})_{u\in\cU}$ and
$\hat{\theta}_{\tau}^{\rm lit}=(\hat{\theta}_{\tau,u}^{\rm lit})_{u\in\cU}$
denote the constrained maximum likelihood estimators of $\theta$ over the regions
$\Theta_{M}^{\circ}$ and $\Theta_{\tau}^{\rm lit}$, respectively (see Section~\ref{section:mle}
for the definition of $\Theta_M^{\circ}$ and $\Theta_{\tau}^{\rm lit}$).
We call $\hat{\theta}_{\tau}^{\rm lit}$ the lasso-type estimator of SGM.
The same notations on the estimators are used also for MixM.

The graphical Gaussian lasso estimator $\hat{C}=\hat{C}(\tau)$ of the concentration matrix (\cite{yuan2007})
is formulated as follows
\[
 \mbox{min.}\quad \{\log\det(C)+\mathop{\rm tr}(\hat{\Sigma} C)\}
 \quad \mbox{s.t.}\quad \sum_{i<j}|C_{ij}|\leq \tau\sum_{i<j}|(\hat{\Sigma}^{-1})_{ij}|,
\]
where $\hat{\Sigma}$ is the sample correlation and the tuning parameter $\tau$ ranges over $[0,1]$.
If $\tau=1$, the graphical Gaussian lasso estimator coincides with the maximum likelihood estimator
(this is not the case for the lasso-type estimators of SGM and MixM).
The partial correlation coefficient of $x_i$ and $x_j$ is estimated by
$\hat{\rho}_{ij}=-\hat{C}_{ij}/\sqrt{\hat{C}_{ii}\hat{C}_{jj}}$.

For given raw data $(D_{ti})_{1\leq t\leq n,1\leq i\leq m}$,
we preprocess it before estimation.
For Gaussian models, we use the data $\tilde{D}_{ti}$
scaled by the standard way:
\[
 \tilde{D}_{ti}\ =\ \frac{D_{ti}-\bar{D}_{\cdot i}}{\mathrm{sd}(D_{\cdot i})},
 \quad \bar{D}_{\cdot i}\ =\ \frac{1}{n}\sum_{t=1}^n D_{ti},
 \quad \mathrm{sd}(D_{\cdot i})\ =\ \sqrt{\frac{1}{n}\sum_{t=1}^n(D_{ti}-\bar{D}_{\cdot i})^2}.
\]
For SGM and MixM,
the data is further transformed into $X_{ti}=\Phi(\tilde{D}_{ti})$,
where $\Phi$ is the standard normal cumulative distribution function,
in order that $X_{ti}$ ranges over $[0,1]$.
By the transform $\Phi$, the standard normal density 
as the null Gaussian model
is transformed into the uniform density as the null SGM and the null MixM.

We used the package SDPT3 for solving the determinant-maximization problem
on MATLAB (\cite{toh2006}).

\subsection{Simulation}

We first confirm that the maximum likelihood estimator
is actually computed by the method described in Section~\ref{section:mle}.
Consider Example~\ref{exmp:complicated} of Subsection~\ref{subsection:examples}.
The true parameter is $\theta_{(1,2,0)}=0.1$, $\theta_{(0,1,1)}=0.3$
and $\theta_{(1,1,1)}=0.2$ with the true frequency set $\cU_0=\{(1,2,0),(0,1,1),(1,1,1)\}$.
The frequency set (\ref{eqn:numerical-freq}) we use for estimation is written in a matrix form
\begin{equation}
 \cU\ =\ \left(
  \begin{array}{cccccccccccccccc}
1&2&0&1&2&0&1&0&1&2&0&1&0&0&1&0\\
0&0&1&1&1&2&2&0&0&0&1&1&2&0&0&1\\
0&0&0&0&0&0&0&1&1&1&1&1&1&2&2&2
  \end{array}
 \right).
 \label{eqn:numerical-freq-3}
\end{equation}
The columns are arranged according to the lexicographic order.
A result of estimation is given in Figure~\ref{fig:exmp-check}.
The sample size is $n=100$ and the number of experiments is $100$.
The samples were generated by the exact method of \cite{sei2006}.
Both estimators actually distribute around the true parameter.

\begin{figure}[htbp]
 \begin{center}
  \begin{tabular}{cc}
   \includegraphics[width=6cm,clip]{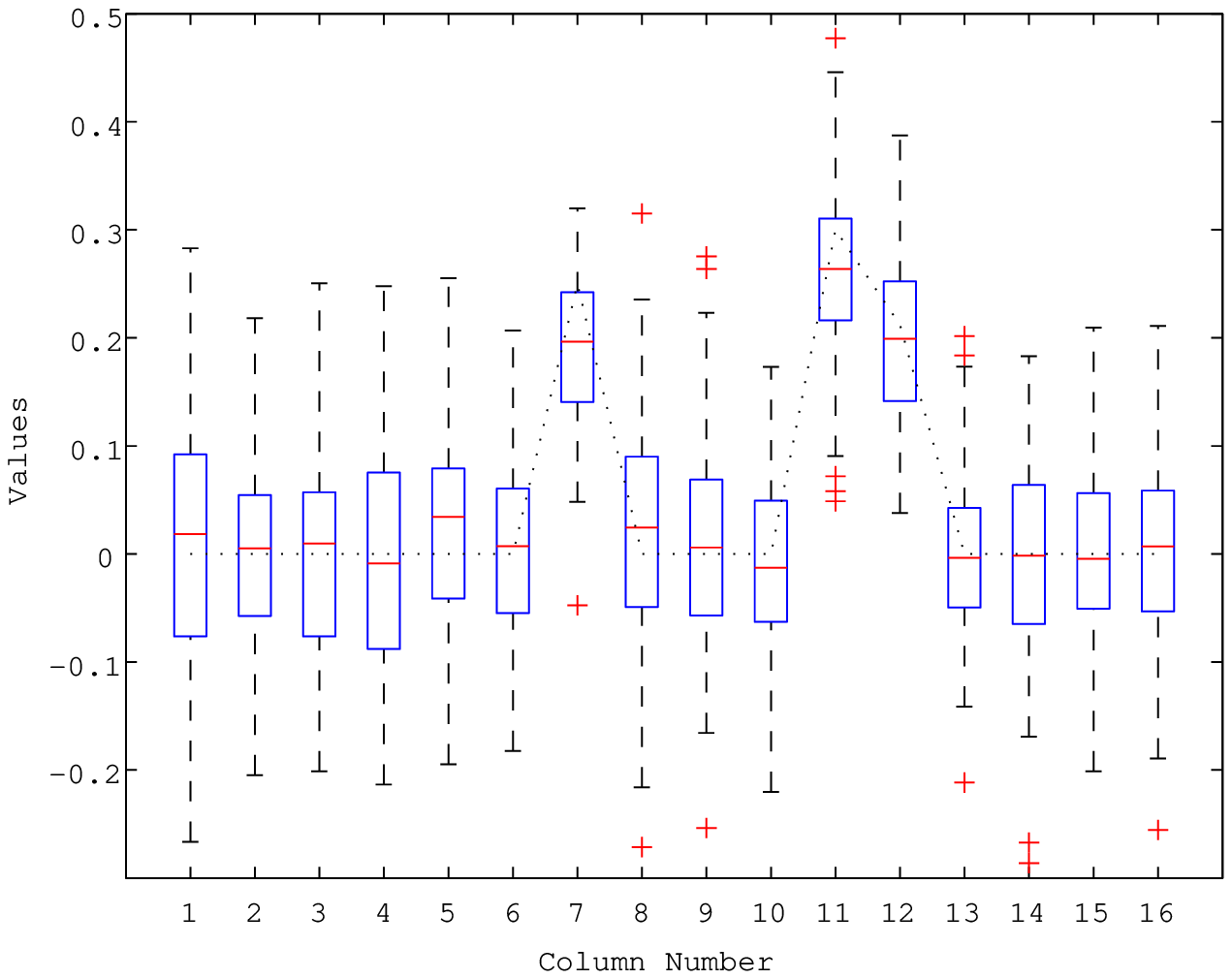}
   &
   \includegraphics[width=6cm,clip]{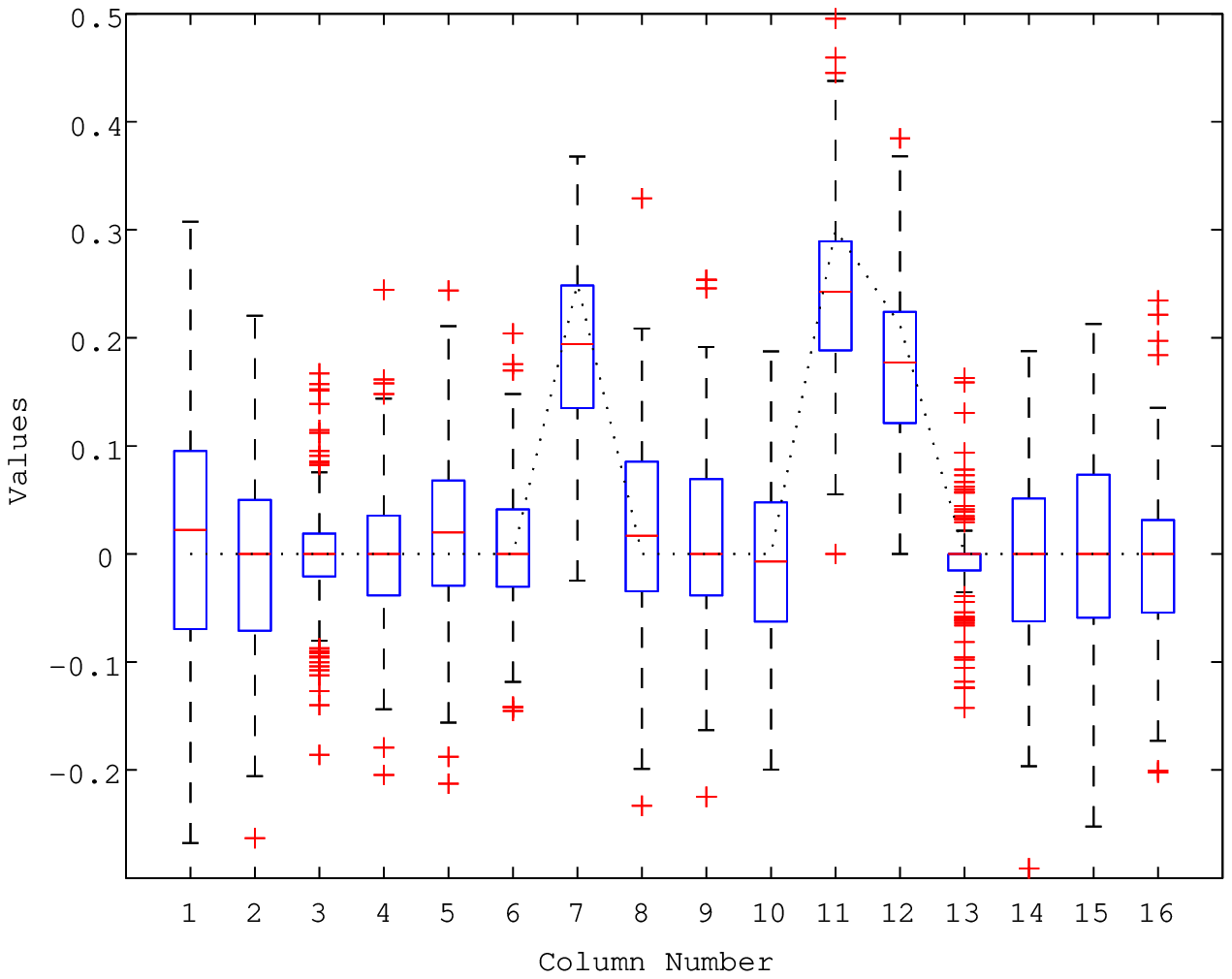}
   \\
   (a) $\sqrt{J_{uu}}\hat{\theta}_{M,u}^{\circ}$ ($M=5$).
   &
   (b) $\sqrt{J_{uu}}\hat{\theta}_{\tau,u}^{\rm lit}$ ($\tau=1$).
  \end{tabular}
 \end{center}
 \caption{A simulation of estimation of SGM.
 The box-plot shows
 each component of the constrained maximum likelihood estimators
 (a) $\hat{\theta}_M^{\circ}$ for $M=5$ and (b) $\hat{\theta}_{\tau}^{\rm lit}$ for $\tau=1$.
 The values are normalized by the square root $\sqrt{J_{uu}}$ of the Fisher information.
 The horizontal axis denotes $u\in\cU$ arranged according to (\ref{eqn:numerical-freq-3}).
 The dashed line denotes the true parameter.
 The sample size is $n=100$ and the number of experiments is $100$.
 }
 \label{fig:exmp-check}
\end{figure}

We next compare SGM with MixM and Gaussian models.
We consider a five-dimensional example.
Let $\phi(x|\mu,\Sigma)$ denote
the normal density with mean $\mu$ and
covariance $\Sigma$.
Let $m=5$ and define the true density $p_0(x)$ by
\begin{eqnarray}
 p_0(x)
 &=& \phi(x_1|0,1)\phi(x_2|x_1,1)\phi(x_3|0,\sigma_3^2(x_2))
 \phi(x_4,x_5|0,\Sigma_{45}(x_3)),
 \label{eqn:5-dim-exmp}
\end{eqnarray}
where
\[
 \sigma_3^2(x_2)\ =\ 1+\tanh(x_2)
 \quad \mbox{and}\quad
 \Sigma_{45}(x_3)\ =\ \left(\begin{array}{cc}1&\tanh(x_3)\\\tanh(x_3)&1\end{array}\right).
\]
By the definition, the set of variables $(x_1,x_2)$
has positive correlation,
the variable $x_3$ has heteroscedasticity against $x_2$,
and the set of variables $(x_3,x_4,x_5)$
has three-dimensional interaction.
Remark that the density does not belong to SGM.
A numerical result is shown in Table~\ref{table:simu}.
The sample size is $n=40$ and
the number of experiments is $200$.
All of the three models detected
the correlation of the pair $(x_1,x_2)$.
However, only SGM effectively detected the heteroscedasticity of $(x_2,x_3)$
and the three-dimensional interaction $(x_3,x_4,x_5)$.
The estimator of MixM was too sparse, and did not effectively detect them.

For the same true density,
we also computed the predictive performance
of the estimators of SGM, MixM and Gaussian.
We use the expected predictive log-likelihood
as the index of the predictive performance.
The arbitrary constant of the log-likelihood is
determined in such a way that the log-likelihood of the null model is zero.
The sample size is $n=40$ for observation
and $10$ for prediction.
The number of experiments is $200$.
The maximum mean predictive log-likelihood of SGM
is estimated as $3.37 (\pm 0.33)$ at $\tau=1.0$,
where the confidence interval is based on the 95\% interval with the normal approximation.
For MixM and Gaussian, the maximum value is estimated as $1.99 (\pm 0.15)$ at $\tau=1.0$
and $2.72 (\pm 0.26)$ at $\tau=0.32$, respectively.
Hence SGM has better predictive performance than MixM and Gaussian.

\subsection{Real dataset}

We consider the digoxin clearance data reported in \cite{halkin1975} (see also \cite{edwards2000}).
The data consists of creatinine clearance ($x_1$),
digoxin clearance ($x_2$) and urine flow ($x_3$) of 35 patients.
In Table~\ref{table:digoxin},
we compare the lasso-type estimators of SGM, MixM and the Gaussian model.
The result shows that for the data our SGM gives slightly
better predictive performance than MixM and the Gaussian models.
As stated in \cite{edwards2000},
partial correlation of $(x_1,x_3)$ is not significant.
However, our model suggests a heteroscedastic effect of 
$x_1$ (creatinine clearance) against $x_3$ (urine flow).

 \begin{table}[htbp]
  \caption{
  Mean value of the lasso-type estimators  for the five-dimensional data.
  The tuning parameter $\tau$ is set to $1$.
  The sample size is $n=40$
  and the number of experiments is $200$.
  The confidence interval is based on the 95\% interval
  with the normal approximation.
  For SGM and MixM, only top ten values of $\sqrt{J}_{uu}\hat{\theta}_{\tau,u}^{\rm lit}$
  are shown.
  For the Gaussian model, $u$ is
  the indicator vector of a pair $(i,j)$.
  }
  \label{table:simu}
  \begin{center}
{\footnotesize
\begin{tabular}{c|r||c|r||c|r}
 \multicolumn{2}{c||}{SGM}
 &\multicolumn{2}{c||}{MixM}
 &\multicolumn{2}{c}{Gaussian}\\
 $u$&\multicolumn{1}{|c||}{$\mathrm{E}[\sqrt{J_{uu}}\hat{\theta}_{\tau,u}^{\rm lit}]$}
&$u$&\multicolumn{1}{|c||}{$\mathrm{E}[\sqrt{J_{uu}}\hat{\theta}_{\tau,u}^{\rm lit}]$}
&$u$&\multicolumn{1}{|c}{$\mathrm{E}[\hat{\rho_{ij}}(\tau)]$}\\
\hline
&&&&\\
 $(1,1,0,0,0)$&0.510 ($\pm 0.013$)
&$(1,1,0,0,0)$&0.123 ($\pm 0.006$)
&$(1,1,0,0,0)$&0.706 ($\pm 0.011$)\\
$(0,0,1,1,1)$&-0.297 ($\pm 0.017$)
&$(0,1,2,0,0)$&-0.031 ($\pm 0.005$)
&$(1,0,0,0,1)$&-0.023 ($\pm 0.023$)\\
$(0,1,2,0,0)$&-0.232 ($\pm 0.015$)
&$(0,0,1,1,1)$&-0.007 ($\pm 0.003$)
&$(0,1,1,0,0)$&0.014 ($\pm 0.023$)\\
$(0,0,2,0,0)$&-0.106 ($\pm 0.014$)
&$(0,0,2,0,0)$&-0.006 ($\pm 0.002$)
&$(1,0,0,1,0)$&-0.010 ($\pm 0.022$)\\
$(2,0,0,0,0)$&-0.095 ($\pm 0.011$)
&$(0,2,0,0,0)$&-0.002 ($\pm 0.001$)
&$(0,1,0,0,1)$&0.008 ($\pm 0.024$)\\
$(0,2,0,0,0)$&-0.084 ($\pm 0.010$)
&$(1,0,2,0,0)$&-0.002 ($\pm 0.001$)
&$(0,0,0,1,1)$&-0.007 ($\pm 0.028$)\\
$(0,0,0,0,2)$&-0.043 ($\pm 0.013$)
&$(2,0,0,0,0)$&-0.001 ($\pm 0.001$)
&$(0,1,0,1,0)$&0.007 ($\pm 0.024$)\\
$(0,0,0,2,0)$&-0.043 ($\pm 0.010$)
&$(0,2,0,1,0)$&-0.000 ($\pm 0.001$)
&$(0,0,1,1,0)$&-0.006 ($\pm 0.023$)\\
$(1,0,2,0,0)$&-0.036 ($\pm 0.009$)
&$(0,0,1,0,2)$&-0.000 ($\pm 0.001$)
&$(1,0,1,0,0)$&-0.004 ($\pm 0.021$)\\
$(0,0,0,2,1)$&-0.015 ($\pm 0.015$)
&$(0,0,0,0,2)$&-0.000 ($\pm 0.001$)
&$(0,0,1,0,1)$&0.004 ($\pm 0.023$)\\
%
\end{tabular}
}
 \end{center}
 \end{table}

\begin{table}[htbp]
\caption{
 A result for the digoxin data.
 The lasso-type estimators of SGM, MixM and the graphical Gaussian model are shown.
 Only non-zero values are displayed.
 For the Gaussian model, the estimated partial correlation of
 the pairs $\{1,2\},\{1,3\},\{2,3\}$ is displayed on the row $u=(1,1,0),(1,0,1),(0,1,1)$, respectively.
 The cross-validated predictive log-likelihood (referred to as CV prediction) is put on the bottom.
 For each model, the asterisk `$*$' indicates the optimal tuning parameter selected by CV prediction.
 }
\label{table:digoxin}
\begin{center}
{\footnotesize
\begin{tabular}{r|rr|rr|rr}
&\multicolumn{2}{|c|}{SGM}
&\multicolumn{2}{|c|}{MixM}
&\multicolumn{2}{|c}{Gaussian}\\
&$\tau=0.5$&$\tau=1.0^*$
&$\tau=0.5$&$\tau=1.0^*$
&$\tau=0.25^*$&$\tau=1.0$\\
\hline
&&&&&&\\
$(1,1,0)$&0.351&0.558&0.177&0.354&0.480&0.758\\
$(0,1,1)$&0.149&0.301&&&0.217&0.485\\
$(2,0,1)$&&-0.166&&&&\\
$(1,0,1)$&0.149&0.148&&&&-0.191\\
$u$\hspace{10pt}
$(0,0,2)$&-0.070&-0.147&&&&\\
$(0,2,0)$&&-0.088&&&&\\
$(1,0,2)$&&0.072&&&&\\
$(0,0,1)$&0.073&0.050&&&&\\
$(0,1,2)$&&-0.039&&&&\\
&&&&&&\\
\hline
CV prediction&
11.19&\underline{14.54}
&6.95&12.26
&14.49&-0.92\\
\end{tabular}
}
\end{center}
\end{table}

\section{Discussion} \label{section:discussion}

We defined SGM as a set of the potential functions $\psi$
and studied its feasible region to calculate the constrained
maximum likelihood estimator.
SGM was applied to both simulated and real dataset.
We discuss remaining mathematical and practical problems.

We used the finite Fourier expansion to define the potential function $\psi$
as Eq.~(\ref{eqn:SGM}).
It is sometimes hard to describe local behavior of the density function
if we use this expansion.
For such purposes,
we can use wavelets instead of the cosine functions
as long as the resultant potential function satisfies the Neumann condition (\ref{eqn:neumann}).
For example, assume that we want to describe tail behavior
of two-dimensional data around $x=(1,1)$.
Then we can use a function
\[
 \psi(x|\theta,a)=(x_1^2+x_2^2)/2+\pi^{-2}\theta(2+\cos(\pi x_1)+\cos(\pi x_2))^{a},
\]
where $a>1/2$.
A typical shape of the density function $p(x|\theta,a)=\det(D^2\psi(x|\theta,a))$
is given in Figure~\ref{fig:tail-dens}.
One can confirm that
the gradient map $D\psi$ is continuous on $[0,1]^2$
and satisfies the Neumann condition~(\ref{eqn:neumann}).
A sufficient condition for convexity of $\psi$ is
$0\leq \theta\leq 2^{1-2a}/a$.
If $a<1$,
then the tail behavior of $p(x|\theta,a)$ is
\[
 p(x|\theta,a)\ \simeq\ \theta^2a^2(2a-1)
 \left(\frac{\pi^2}{2}\{(1-x_1)^2+(1-x_2)^2\}\right)^{2(a-1)}
\]
as $(x_1,x_2)\to (1,1)$.
The proofs of these facts are omitted.
Although estimation of $\theta$ is described by the determinant maximization,
that of $a$ is not.
Further investigation is needed.

\begin{figure}[htbp]
 \begin{center}
  \includegraphics[width=6cm,clip]{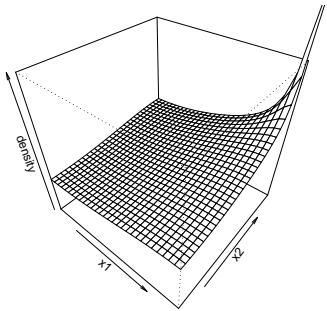}
 \end{center}
 \caption{The density function $p(x|\theta,a)$ for
 $a=0.75$ and $\theta=2^{1-2a}/a$.}
 \label{fig:tail-dens}
\end{figure}

If any covariates are available together with given data,
we can include the covariates
in the parameter $\theta$ of SGM.
However, since the parameter space $\Theta$ of SGM is not the whole Euclidean space,
its use is restricted.

The author recently proved an inequality on
Efron's statistical curvature,
in that the curvature of SGM at the origin $\theta=0$
is always smaller than that of MixM (\ref{eqn:mixture}).
This fact is not so practical
but it supports SGM.
Since the statement and the proof of this inequality are rather complicated,
we will present them in a forthcoming paper.

We constructed a lasso-type estimator on SGM
as a byproduct of the conservative feasible region in Section~\ref{section:mle}.
Performance of the estimator is numerically studied in Section~\ref{section:numerical}.
For the existing lasso estimators,
some asymptotic results are known
when the sample size $n$ and/or the number $m$ of variates increase
 (\cite{knight2000}, 
\cite{meinshausen2006}, 
\cite{yuan2007}, 
\cite{bunea2007}, 
\cite{banerjee2008}).
We think it is important to compare our SGM
with the Gaussian, mixture and exponential models
on the asymptotic argument.

\appendix

\section{Proofs}

\subsection{Proof of Lemma~\ref{lem:fundamental}}

 Let $\psi$ have the form (\ref{eqn:SGM}) and
 choose any $\theta$ such that $D^2\psi(x|\theta)\succeq 0$ for every $x\in[0,1]^m$.
 We prove that the gradient map $D\psi(\cdot|\theta)$ is a bijection on $[0,1]^m$.
 If $\theta=0$, then the bijectivity of $D\psi(x|\theta)=x$ is clear.
 Therefore we assume $\theta\neq 0$.
 We can extend the domain of $\psi(\cdot|\theta)$ from $[0,1]^m$ to whole $\mbR^m$
 by using Eq.~(\ref{eqn:SGM}),
 and denote the extended function by $\tilde{\psi}(x)=\tilde{\psi}(x|\theta)$ for $x\in\mbR^m$.
 Since $\tilde{\psi}(x)$ is a periodic and even function
 along each axis, the convexity condition $D^2\tilde{\psi}\succeq 0$
 holds over $x\in\mbR^m$.
 We will prove that
 (i) $D\tilde{\psi}$ is a bijection on $\mbR^m$
 and (ii) $D\tilde{\psi}$ is a bijection on each hyperplane $\{x\mid x_j=b\}$, where
 $j\in\{1,\ldots,m\}$ and $b\in\{0,1\}$.
 We first show that the bijectivity on $[0,1]^m$ follows from the conditions (i) and (ii).
 Indeed, if (i) and (ii) are fulfilled, then for each $j\in\{1,\ldots,m\}$
 the sandwiched region $\{x\in\mbR^m\mid 0\leq x_j\leq 1\}$ between two hyperplanes
 is mapped onto itself because $D\tilde{\psi}$ is continuous.
 Therefore $[0,1]^m$ is injectively mapped onto itself.
 To prove (i), it is sufficient to show that $\tilde{\psi}$ is strictly convex and co-finite:
 $\lim_{\lambda\to\infty}\tilde{\psi}(\lambda x)/\|x\|=0$
 whenever $x\neq 0$
 (see Theorem 26.6 of \cite{rockafeller1970}).
 We define a function $f(z)$ of $z\in\mbR$
 by $f(z)=\tilde{\psi}(x_0+ze)$,
 where $x_0\in\mbR^m$ and $e\in\mbR^m\setminus\{0\}$ are arbitrary.
 Then $f''(z)\geq 0$ for any $z$
 since $D^2\tilde{\psi}(x)\succeq 0$ for any $x\in\mbR^m$.
 However, since $f''(z)$ is a non-constant analyitc function (recall that $\theta\neq 0$),
 $f''(z)$ must be positive except for a finite number of $z$ for each bounded interval.
 Hence $f$, and therefore $\tilde{\psi}$, is strictly convex.
 The co-finiteness of $\tilde{\psi}$
 is immediate because $\tilde{\psi}$ is sum of $x^{\top}x/2$
 and a bounded function.
 Hence (i) was proved.
 Next we prove the condition (ii).
 We consider the hyperplane $\{x\mid x_m=b\}$, where $b\in\{0,1\}$,
 without loss of generality.
 Denote the restriction of $\tilde{\psi}$ to $\{x\mid x_m=b\}$
 by $\tilde{\psi}_{m-1}$.
 Then $\tilde{\psi}_{m-1}$ has the following expression
 \[
 \tilde{\psi}_{m-1}(x_1,\ldots,x_{m-1})
  \ =\ \frac{b^2}{2} + \frac{1}{2}\sum_{i=1}^{m-1}x_i^2
  - \sum_{u\in\cU}\pi^{-2}\theta_u(-1)^{u_jb}\prod_{i=1}^{m-1}\cos(\pi u_jx_j).
 \]
 This function is the same form as Eq.~(\ref{eqn:SGM}) with the dimension $m-1$.
 The convexity condition $(\Rd^2\tilde{\psi}_{m-1}/\Rd x_i\Rd x_j)\succeq 0$
 is also satisfied because $\tilde{\psi}_{m-1}$ is a restriction of $\tilde{\psi}$.
 Thus (ii) is proved in the same manner as the proof of (i).

\subsection{Proof of Lemma~\ref{lem:SGM-is-not}}

A statistical model is a mixture model if and only if
all the second derivatives of the density function
with respect to the parameter vanish.
Hence we calculate the second derivative of the density function of SGM.
Put $\mbZ_i:=\{u\in\mbZ_{\geq 0}^m\mid u_j=0\ \forall j\neq i\}$.
If $\cU\subset\mbZ_i$
for some $i$,
then it is easy to confirm that SGM becomes a mixture model
\[
 p(x|\theta)\ =\ 1+\sum_{u\in\cU}\theta_uu_i^2\cos(\pi u_ix_i).
\]
Hence we assume that $\cU\not\subset\mbZ_i$ for any $i$.
Then there exist $u,v\in\cU$ (the case $u=v$ is available) such that $|\sigma(u)\cup\sigma(v)|\geq 2$,
where $\sigma(u)=\{j\mid u_j>0\}$.
Putting $A_u=\{D^2\psi(x|\theta)\}^{-1}\{\Rd/\Rd\theta_u(D^2\psi(x|\theta))\}$ we have
\[
 \frac{\Rd^2p(x|\theta)}{\Rd\theta_u\Rd\theta_v}
  \ =\ \mathop{\rm tr}A_u\mathop{\rm tr}A_v-\mathop{\rm tr}[A_uA_v].
\]
Since $A_u|_{\theta=0,x=0}=\mathrm{diag}(u_1^2,\ldots,u_m^2)$, we have
\[
 \left.\frac{\Rd^2p(x|\theta)}{\Rd \theta_u\Rd \theta_v}\right|_{\theta=0,x=0}
  \ =\ \|u\|^2\|v\|^2-\sum_{i}u_i^2v_i^2
  \ =\ \sum_{i}\sum_{j\neq i}u_i^2v_j^2
  \ >\ 0,
\]
where the last inequality follows from $|\sigma(u)\cup\sigma(v)|\geq 2$.
Thus SGM is not a mixture model
as long as $\cU\not\subset\mbZ_i$ for any $i$.

\subsection{Proof of Lemma\ref{lem:SGM-and-MixM}}

The score function of SGM at $\theta=0$ is directly calculated as
\[
 L_u
 \ :=\ \left.\frac{\Rd}{\Rd\theta_u}\log p(x|\theta)\right|_{\theta=0}
 \ =\ \|u\|^2\prod_{j=1}^m\cos(\pi u_jx_j).
\]
The score function of MixM is also easily
proved to be $L_u$.
Then the Fisher information matrix of both the models is
\[
 J_{uv}
 \ =\ \int p(x|0)L_uL_v\mathrm{d}x
 \ =\ \|u\|^2\|v\|^2
 \prod_{j=1}^m \int_0^1 \cos(\pi u_jx_j)\cos(\pi v_jx_j)\mathrm{d}x_j.
\]
Here the integral is calculated by the following formula
\[
 \int_0^1\cos(\pi u_jx_j)\cos(\pi v_jx_j)\mathrm{d}x_j
 \ =\ \left\{
 \begin{array}{ll}
  1& \mbox{if}\ u_j=v_j=0,\\
  1/2& \mbox{if}\ u_j=v_j>0,\\
  0& \mbox{if}\ u_j\neq v_j.
 \end{array}
 \right.
\]

\subsection{Proof of Equations (\ref{eqn:Fisher-explicit-1}) and (\ref{eqn:Fisher-explicit-2})}

We first prove Eq.~(\ref{eqn:Fisher-explicit-1}).
Let $m=1$ and $\cU=\{u\}$.
We only consider the case $u=1$.
The other cases are similarly proved.
Put $\theta=\theta_{1}$.
Since $p(x_1|\theta)=1+\theta \cos(\pi x_1)$,
we have
\[
 J_{uu}(\theta)
 \ =\ \int_0^1\frac{\cos^2(\pi x_1)}{1+\theta \cos(\pi x_1)}\mathrm{d}x_1.
\]
By putting $z=\exp(\mathrm{i}\pi ux_1)$, we obtain
\[
 J_{uu}(\theta)
 \ =\ \frac{1}{2\pi\mathrm{i}}
  \oint_{|z|=1}\frac{(z+z^{-1})^2/4}{1+\theta(z+z^{-1})/2}\frac{\mathrm{d}z}{z}
 \ =\ \frac{1}{4\pi\mathrm{i}}
  \oint_{|z|=1}\frac{(z^2+1)^2}{z^2(\theta z^2+2z+\theta)}\mathrm{d}z.
\]
The poles of the integrand inside the unit circle are $0$ and $z_+$,
where $z_{\pm}:=(-1\pm\sqrt{1-\theta^2})/\theta$.
By the residue theorem, we obtain
\begin{eqnarray*}
 J_{uu}(\theta)
 \ =\ \frac{1}{2}\left(\frac{-2}{\theta^2}\right)
  +\frac{1}{2}
  \frac{(z_+^2+1)^2}{z_+^2\theta(z_+-z_-)}
 \ =\ \frac{1-\sqrt{1-\theta^2}}{\theta^2\sqrt{1-\theta^2}}.
\end{eqnarray*}
This proves Eq.~(\ref{eqn:Fisher-explicit-1}).

We next prove Eq.~(\ref{eqn:Fisher-explicit-2}).
Put $u=(1,1)$ and $\theta=\theta_u$.
We use the following identity
\begin{eqnarray*}
 p(x|\theta)
 &=& \det\left(
  \begin{array}{cc}
   1+\theta\cos(x_1)\cos(x_2)& -\theta\sin(x_1)\sin(x_2) \\
  -\theta\sin(x_1)\sin(x_2)& 1+\theta\cos(x_1)\cos(x_2)
  \end{array}
 \right)
  \\
 &=& (1+\theta\cos(\pi (x_1-x_2)))(1+\theta\cos(\pi (x_1+x_2))).
\end{eqnarray*}
The Fisher information is
\begin{eqnarray*}
 J_{uu}(\theta)
  &=& \int_{[0,1]^2}\left(
   \frac{\cos^2(\pi (x_1-x_2))}{1+\theta\cos(\pi (x_1-x_2))}
   + \frac{\cos^2(\pi (x_1+x_2))}{1+\theta\cos(\pi (x_1+x_2))}
   \right)
 \mathrm{d}x_1\mathrm{d}x_2
   \\
  &=& \frac{1}{4}\int_{[-1,1]^2}\left(
   \frac{\cos^2(\pi (x_1-x_2))}{1+\theta\cos(\pi (x_1-x_2))}
   + \frac{\cos^2(\pi (x_1+x_2))}{1+\theta\cos(\pi (x_1+x_2))}
   \right)
 \mathrm{d}x_1\mathrm{d}x_2
   \\
  &=& \frac{1}{4}\int_{[-1,1]^2}
  \left(
   \frac{\cos^2(\pi y_1)}{1+\theta\cos(\pi y_1)}
   +\frac{\cos^2(\pi y_2)}{1+\theta\cos(\pi y_2)}
  \right)
 \mathrm{d}y_1\mathrm{d}y_2
\end{eqnarray*}
where the last equality follows from
the transformation $y_1=x_1-x_2$ and $y_2=x_1+x_2$,
and from the periodicity of the integrand.
Then (\ref{eqn:Fisher-explicit-2}) is proved in the same manner as the proof of
(\ref{eqn:Fisher-explicit-1}).

\subsection{Proof of Lemma~\ref{lem:interior}}

We use the following elementary lemma.
Put $\cS=\{A\succeq 0\mid \mathop{\rm tr}A=1\}$.
Note that $\cS$ is compact.

\begin{lem} \label{lem:positive-definite-cond}
 Let $X$ be a real symmetric matrix.
 Then the minimum eigenvalue of $X$ is given by
 $\min_{A\in\cS}\mathop{\rm tr}(AX)$.
\end{lem}

\begin{proof}
 Let $X=\sum_{i}\xi_ie(i)e(i)^{\top}$
 be the spectral decomposition of $X$,
 where $\xi_1\leq\cdots\leq\xi_m$ and $e(i)^{\top}e(i)=1$.
 For any $A\in\cS$,
 \[
  \mathop{\rm tr}(AX)
 \ =\ \sum_{i}\xi_i(e(i)^{\top}Ae(i))
 \ \geq\ \xi_1\sum_{j}(e(j)^{\top}Ae(j))
 \ =\ \xi_1.
 \]
 The equality is attained at $A=e(1)e(1)^{\top}$.
\end{proof}

Let $H_u(x)=D^2(-\pi^{-2}\prod_{j=1}^m\cos(\pi u_jx_j))$.
Then
\[
 D^2\psi(x|\theta)\ =\ I+\sum_{u\in\cU}\theta_u H_u(x).
\]
The minimum eigenvalue $\rho_{\min}(\theta)$ of $D^2\psi(x|\theta)$
minimized over $x\in[0,1]^m$ is
\begin{eqnarray*}
 \rho_{\min}(\theta)
 &=& 1+\min_{x\in [0,1]^m,A\in\cS}\sum_{u\in\cU}\theta_u\mathop{\rm tr}(AH_u(x)).
\end{eqnarray*}
Recall that the parameter space $\Theta$
is expressed as $\Theta=\{\theta\in\mbR^{\cU}\mid \rho_{\min}(\theta)\geq 0\}$.
We prove that the interior of $\Theta$ is $\Theta^{\circ}=\{\theta\in\mbR^{\cU}\mid \rho_{\min}(\theta)>0\}$.
Put
\[
 \mu\ =\ \max_{u\in\cU}\max_{x\in[0,1]^m}\max_{A\in\cS}|\mathop{\rm tr}(AH_u(x))|
 \ <\ \infty.
\]
We first prove that if $\rho_{\min}(\theta)>0$, then $\theta\in\Theta^{\circ}$.
Indeed, if $\eta\in\mbR^{\cU}$ is sufficiently small, then
\[
 \rho_{\min}(\theta+\eta)
 \ \geq\ \rho_{\min}(\theta)-\mu\sum_{u\in\cU}|\eta_u|
 \ \geq\ 0.
\]
We next prove that if $\rho_{\min}(\theta)=0$, then $\theta\in\Theta\setminus\Theta^{\circ}$.
Since $\rho_{\min}(\theta)=0$,
there exist some $A\in\cS$ and some $x\in[0,1]^m$
such that $\mathop{\rm tr}(AD^2\psi(x|\theta))=0$.
For such an $x$,
there exists some $v\in\cU$ such that
$\theta_v\mathop{\rm tr}(AH_v(x))<0$.
Define a vector $\eta\in\mbR^{\cU}$ by $\eta_u=\theta_v 1_{\{u=v\}}$.
Then, for any $\epsilon>0$, we have
\[
 \rho_{\min}(\theta+\epsilon\eta)
 \ \leq\ \mathop{\rm tr}(AD^2\psi(x|\theta+\epsilon\eta))
 \ =\ \epsilon\theta_v\mathop{\rm tr}(AH_v(x))
 \ <\ 0.
\]
This implies that $\theta$ is a boundary point of $\Theta$.
Hence Lemma~\ref{lem:interior} was proved.

\subsection{Proof of Theorem~\ref{thm:inner-approx}}

We first recall some notations.
We use $[m]=\{1,\ldots,m\}$ and $L_M=\{\frac{0}{M},\frac{1}{M},\cdots,\frac{M}{M}\}$.
The supremum norm of $s\in\mbZ^m$ is defined by $\|s\|_{\infty}:=\max_{j}|s_j|$.
Recall that $U_{\max}=\max_{u\in\cU}\|u\|_{\infty}$.
We denote $U=U_{\max}$ for simplicity.
Recall that $K_M$ is a linear map on $\mbR^{\cU}$
defined by $K_M\theta=(\theta_u/\prod_{j=1}^m(1-u_j/M))_{u\in\cU}$.

Define a set $K_M^{-1}\Theta^{\circ}$ by
\[
 K_M^{-1}\Theta^{\circ}
 \ :=\ \{K_M^{-1}\theta\mid \theta\in\Theta^{\circ}\}
 \ =\ \{\theta\mid D^2\psi(x|K_M\theta)\succ 0\quad \forall x\in[0,1]^m\}.
\]
Then we have $K_M^{-1}\Theta^{\circ}\ \subset\ \Theta_M^{\circ}$
by the definition of $\Theta_M^{\circ}$.
Hence, the theorem follows from the following two claims.
\begin{itemize}
 \item[(i)] $\displaystyle\limsup_{M\to\infty}K_M^{-1}\Theta^{\circ}=\Theta^{\circ}$.
 \item[(ii)] $\Theta_M^{\circ}\subset\Theta^{\circ}$ for any $M$.
\end{itemize}

We first prove (i).
Put $\cS=\{A\succeq 0\mid \mathop{\rm tr}A=1\}$
and $f(x|\theta,A)=\mathop{\rm tr}[AD^2\psi(x|\theta)]$.
By Lemma\ref{lem:positive-definite-cond} and compactness of $[0,1]^m\times \cS$,
a vector $\theta$ belongs to $\Theta^{\circ}$ if and only if
\[
 \min_{x\in [0,1]^m,A\in\cS}f(x|\theta,A)\ >\ 0.
\]
Now it is sufficient to prove that, for any $\theta\in\mbR^{\cU}$, $f(x|K_M\theta,A)$
converges to $f(x|\theta,A)$ uniformly in $x\in[0,1]^m$ and $A\in\cS$.
Let $H_u(x):=D^2(-\pi^{-2}\prod_{j=1}^m\cos(\pi u_jx_j))$.
Then we have $f(x|\theta,A)=1+\sum_{u\in\cU}\theta_u\mathop{\rm tr}[AH_u(x)]$
and therefore
\begin{eqnarray}
  |f(x|K_M\theta,A)-f(x|\theta,A)|
 &\leq&
 \sum_{u\in\cU}
 \left|\{(K_M\theta)_u-\theta_u\}\mathop{\rm tr}[AH_u(x)]\right|.
 \label{eqn:inner-theorem-proof-1}
 \end{eqnarray}
 Since the function $\mathop{\rm tr}[AH_u(x)]$ of $(x,A)\in[0,1]^m\times \cS$
 is bounded and since $(K_M\theta)_u$ converges to $\theta_u$ for each $u\in\cU$
 as $M\to\infty$,
 the right hand side of (\ref{eqn:inner-theorem-proof-1}) converges to $0$
 uniformly in $x$ and $A$.

 Next we prove (ii).
 Let $R_M=\{-\frac{M-1}{M},\ldots,\frac{M-1}{M},\frac{M}{M}\}$.
 We extend the domain of $\psi$ from $[0,1]^m$ to $\mbR^m$
 as done in the proof of Lemma~\ref{lem:fundamental},
 and denote it again by $\psi$.
 If $\theta\in\Theta_M^{\circ}$,
 then $D^2\psi(\xi|K_M\theta)$
 is positive definite for any $\xi\in R_M^m$
 because $\psi(x|\theta)$ is an even function
 with respect to each coordinate $x_j$.
 Then it is sufficient to prove that
 $D^2\psi(x|\theta)$ for any $x$
 is written as a convex combination of
 $\{D^2\psi(\xi|K_M\theta)\}_{\xi\in R_M^m}$.
 Define a Fej\'er-type kernel $Q_M$ by
 \[
  Q_M(z)\ =\ \frac{1}{2M^2}\sum_{a=0}^{M-1}\sum_{b=0}^{M-1}\mathrm{e}^{\mathrm{i}\pi (a-b)z}
  \ =\ \frac{1}{2M^2}\left(\frac{\sin(\pi Mz/2)}{\sin(\pi z/2)}\right)^2.
 \]
 Then the following lemma holds.

\begin{lem}
 For any $M\geq U+1$, we have
 \[
 D^2\psi(x|\theta)\ =\ 
 \sum_{\xi\in R_M^m}D^2\psi(\xi|K_M\theta)
 \prod_{j=1}^m Q_M(x_j-\xi_j).
 \]
 The right hand side
 is a convex combination of
 $\{D^2\psi(\xi|K_M\theta)\}_{\xi\in R_M^m}$.
\end{lem}

\begin{proof}
 For each $j\in\{1,\ldots,m\}$,
 define an operator $K_{M,j}$ on $\mbR^{\cU}$
 by 
 \[
 (K_{M,j}\theta)_u=\frac{\theta_u}{1-u_j/M}.
 \]
 Then we have $K_M=\prod_{j=1}^mK_{M,j}$ from the definition.
 It is sufficient to show that
 \begin{equation}
  D^2\psi(x|\theta)
 \ =\ \sum_{\xi_j\in R_M}D^2\psi(\xi_j,x_{\setminus j}|K_{M,j}\theta)
 Q_M(x_j-\xi_j),
 \label{eqn:inner-theorem-proof-2}
 \end{equation}
 where $x_{\setminus j}=(x_1,\ldots,x_{j-1},x_{j+1},\ldots,x_m)$.
 In fact, if (\ref{eqn:inner-theorem-proof-2}) is proved,
 then
 \begin{eqnarray*}
  D^2\psi(x|\theta)
   &=& \sum_{\xi_1\in R_M}
   D^2\psi(\xi_1,x_2,\ldots,x_m|K_{M,1}\theta)
   Q_M(x_1-\xi_1)
    \\
   &=& \sum_{\xi_1\in R_M}
   \sum_{\xi_2\in R_M}
   D^2\psi(\xi_1,\xi_2,\ldots,x_M|K_{M,1}K_{M,2}\theta)
   \prod_{j=1}^2Q_M(x_j-\xi_j)
    \\
   &=& \cdots
    \\
   &=& \sum_{\xi\in R_M^m}
   D^2\psi(\xi|K_M\theta)
   \prod_{j=1}^m Q_M(x_j-\xi_j).
 \end{eqnarray*}
 We prove (\ref{eqn:inner-theorem-proof-2})
 for $j=1$ without loss of generality.
 We first describe $D^2\psi(x|\theta)$
 in terms of $\{\mathrm{e}^{\mathrm{i}\pi s^{\top}x}\}_{s\in\mbZ^m}$.
 For each $s\in\mbZ^m$, we define a $m\times m$ matrix
 \[
  F_s
  \ =\ \left\{
  \begin{array}{ll}
   I& \mbox{if}\ s=0,
    \\
   \theta_u2^{-|\sigma(u)|}ss^{\top}
    & \mbox{if}\ |s_j|=u_j\ \mbox{for\ all}\ j\in[m]\ \mbox{for\ some}\ u\in\cU,
    \\
   0& \mbox{otherwise}.
  \end{array}
  \right.
 \]
 Recall that $\sigma(u)=\{j\in[m]\mid u_j>0\}$.
 Then,
 by applying the Euler's formula $\cos(\pi u_jx_j)=(\mathrm{e}^{\mathrm{i}\pi u_jx_j}-\mathrm{e}^{-\mathrm{i}\pi u_jx_j})/2$
 to Eq.~(\ref{eqn:SGM}),
 we can show that
 \[
 D^2\psi(x|\theta)
 \ =\ \sum_{\|s\|_{\infty}\leq U}F_s\mathrm{e}^{\mathrm{i}\pi s^{\top}x}.
 \]
 Recall that $U=\max_{u\in\cU}\|u\|_{\infty}$.
 The right hand side of (\ref{eqn:inner-theorem-proof-2}) with $j=1$ is
 \begin{eqnarray*}
 \lefteqn{
  \sum_{\xi_1\in R_M}D^2\psi(\xi_1,x_{\setminus 1}|K_{M,1}\theta)Q_M(x_1-\xi_1)
 }\\
   &=&
  \sum_{\xi_1\in R_M}
  \left(
   \sum_{\|s\|_{\infty}\leq U}\frac{F_{s}\mathrm{e}^{\mathrm{i}\pi (s_1\xi_1+s_{\setminus 1}^{\top}x_{\setminus 1})}}{1-|s_1|/M}
  \right)
  \left(
   \frac{1}{2M^2}
   \sum_{a=0}^{M-1}\sum_{b=0}^{M-1}\mathrm{e}^{\mathrm{i}\pi (a-b)(x_1-\xi_1)}
  \right)
  \\
  &=&
  \sum_{a=0}^{M-1}\sum_{b=0}^{M-1}
  \sum_{\|s\|_{\infty}\leq U}
  \frac{F_s\mathrm{e}^{\mathrm{i}\pi ((a-b)x_1+s_{\setminus 1}^{\top}x_{\setminus 1})}}{M-|s_1|}
  \ \frac{1}{2M}
  \sum_{\xi_1\in R_M}\mathrm{e}^{\mathrm{i}\pi (s_1-a+b)\xi_1}
  \\
  &=&
  \sum_{a=0}^{M-1}\sum_{b=0}^{M-1}\sum_{\|s\|_{\infty}\leq U}
  \frac{F_s\mathrm{e}^{\mathrm{i}\pi s^{\top}x}}{M-|s_1|}
  1_{\{s_1\equiv a-b\ \mathrm{mod}\ 2M\}}.
 \end{eqnarray*}
 For any $s_1$ with $|s_1|\leq U<M$,
 the cardinality of the set
 \[
  \{(a,b)\in\{0,\ldots,M-1\}^2\mid s_1=a-b\}
 \]
 is $M-|s_1|$.
 Hence we have
 \[
  \sum_{\xi_1\in R_M}D^2\psi(\xi_1,x_{\setminus 1}|K_{M,1}\theta)Q_M(x_1-\xi_1)
 \ =\ \sum_{\|s\|_{\infty}\leq U} F_s\mathrm{e}^{\mathrm{i}\pi s^{\top}x}
 \ =\ D^2\psi(x|\theta).
 \]
 Therefore (\ref{eqn:inner-theorem-proof-2}) was proved.

 Now we prove that
 $\{\prod_{j=1}^mQ_M(x_j-\xi_j)\}_{\xi\in R_M^m}$
 becomes a probability vector.
 In fact, non-negativity follows from the definition of $Q_M$
 and the total mass is $1$ because
 \[
  \sum_{\xi_1\in R_M}Q_M(x_1-\xi_1)
  \ =\ \frac{1}{2M^2}\sum_{a=0}^{M-1}\sum_{b=0}^{M-1}\sum_{\xi_1\in R_M}
  \mathrm{e}^{\mathrm{i}\pi (a-b)(x_1-\xi_1)}
 \ =\ \frac{1}{M}\sum_{a=0}^{M-1}\sum_{b=0}^{M-1}1_{\{a=b\}}
 \ =\ 1.
 \]
 Therefore the lemma and Theorem~\ref{thm:inner-approx} are proved.
\end{proof}

\subsection{Proof of Theorem~\ref{thm:little}}

 Let $\theta\in\Theta^{\rm lit}$.
 We show that $D^2\psi(x|\theta)\succeq 0$ for all $x\in[0,1]^m$.
 By Euler's formula, we obtain
 \[
  \prod_{j=1}^m\cos(\pi u_jx_j)
 \ =\ 2^{-m}\sum_{\alpha\in\{-1,1\}^m}\cos(\pi \alpha^{\top}d(u)x),
 \]
 where $d(u)$ is the $m\times m$ diagonal matrix with the diagonal vector $u$.
 Note that
 $2^{-m}\sum_{\alpha\in\{-1,1\}^m}\alpha\alpha^{\top}=I$.
 Then
\begin{eqnarray*}
 D^2\psi(x|\theta)
  &=& I+
   \sum_{u\in\cU}\frac{\theta_u}{2^m}\sum_{\alpha\in\{-1,1\}^m}
    \cos(\pi \alpha^{\top}d(u)x)d(u)\alpha\alpha^{\top}d(u)
   \\
  &\succeq& I - \sum_{u\in\cU}\frac{|\theta_u|}{2^m}
  \sum_{\alpha\in\{-1,1\}^m} d(u)\alpha\alpha^{\top}d(u)
   \\
  &=& I - \sum_{u\in\cU}|\theta_u|d(u)^2
   \\
  &\succeq& 0.
\end{eqnarray*}
 This implies that $\theta\in\Theta$.

 Next we assume that $\cV\subset\cU$ is linearly independent modulo 2.
 Since $\Theta^{\rm lit}\subset\Theta$, it is sufficient to prove that $\Theta\cap\mbR_{\cV}\subset\Theta^{\rm lit}\cap\mbR_{\cV}$.
 Let $\theta\in\Theta\cap\mbR_{\cV}$.
 We evaluate $D^2\psi(x|\theta)$ at lattice points $\xi\in\{0,1\}^m$.
 For any $\xi\in\{0,1\}^m$ and any $v\in \mbZ^m$, we have
 \[
 \left.D^2\left(-\pi^{-2}\prod_{j=1}^m\cos(\pi v_jx_j)\right)\right|_{x=\xi}
 = (-1)^{v^{\top}\xi}d(v)^2.
 \]
 Since $\cV$ is linearly independent modulo 2,
 we can choose $\xi\in\{0,1\}^m$ such that $v^{\top}\xi=1_{\{\theta_v>0\}}$ (mod $2$)
 for all $v\in\cV$.
 Then
 \begin{eqnarray*}
  0\ \preceq\ \left.D^2\psi(x|\theta)\right|_{x=\xi}
 \ =\ 1 + \sum_{v\in\cV}\theta_v(-1)^{v^{\top}\xi}d(v)^2
 \ =\ 1 - \sum_{v\in\cU}|\theta_v|d(v)^2.
 \end{eqnarray*}
 This means $\theta\in\Theta^{\rm lit}\cap\mbR_{\cV}$.

%

\section*{Acknowledgements}

This study was partially supported by the Global Center of
Excellence ``The research and training center for new development in
mathematics''
and by the Ministry of Education, Science, Sports and Culture, Grant-in-Aid
for Young Scientists (B), No.~19700258.



\end{document}